\documentclass[ss]{imsart}

\RequirePackage{amsthm,amsmath,amsfonts,amssymb}
\RequirePackage[numbers]{natbib}
\RequirePackage[colorlinks,citecolor=blue,urlcolor=blue]{hyperref}
\RequirePackage{graphicx}

\usepackage{stmaryrd,stackrel,algorithmic}
\usepackage{float}

\usepackage{caption}
\usepackage{subcaption}

\usepackage[ruled,vlined]{algorithm2e}

\startlocaldefs

\renewcommand{\le}{\leqslant}
\renewcommand{\leq}{\leqslant}
\renewcommand{\ge}{\geqslant}
\renewcommand{\geq}{\geqslant}
\theoremstyle{plain}

\newtheorem{theorem}{Theorem}[section]

\theoremstyle{definition}
\newtheorem{definition}[theorem]{Definition}

\theoremstyle{remark}

\endlocaldefs

\begin{document}
\begin{frontmatter}
\title{Understanding the Hamiltonian Monte Carlo through its Physics Fundamentals and Examples}
\runtitle{Hamiltonian Monte Carlo}

\begin{aug}
\author[A]{\fnms{Mario}~\snm{Molina}\ead[label=e1]{mariomolina@comunidad.unam.mx}}
\author[B]{\fnms{Abraham}~\snm{Granados}\ead[label=e2]{abrahamgc@uchicago.edu}},
\author[A]{\fnms{Jorge}~\snm{González~Cázares}\ead[label=e3]{jorge.gonzalez@sigma.iimas.unam.mx}}
\and
\author[C]{\fnms{Isaías}~\snm{Bañales}\ead[label=e4]{isaias@ciencias.unam.mx}},

\address[A]{Department of Probability and Statistics,
IIMAS--UNAM\printead[presep={,\ }]{e1,e3}}

\address[B]{Department of Mathematics,
University of Chicago\printead[presep={,\ }]{e2}}

\address[C]{Faculty of Science,
UNAM\printead[presep={,\ }]{e4}}

\runauthor{M. Molina, A. Granados, J. González~Cázares and I. Bañales}
\end{aug}

\begin{abstract}
The Hamiltonian Monte Carlo (HMC) algorithm is a powerful Markov Chain Monte Carlo (MCMC) method that uses Hamiltonian dynamics to generate samples from a target distribution. To fully exploit its potential, we must understand how Hamiltonian dynamics work and why they can be used in a MCMC algorithm. This work elucidates the Monte Carlo Hamiltonian, providing comprehensive explanations of the underlying physical concepts. It is intended for readers with a solid foundation in mathematics who may lack familiarity with specific physical concepts, such as those related to Hamiltonian dynamics. Additionally, we provide  Python code for the HMC algorithm, examples and comparisons with the Random Walk Metropolis-Hastings (RWMH) algorithm, alongside an exploration of modern variants such as the No-U-Turn Sampler (NUTS), the Dual Averaging (DA) scheme and the repelling-attracting HMC (raHMC), to highlight HMC's strengths and weaknesses when applied to Bayesian Inference.
\end{abstract}



\end{frontmatter}

\section{Introduction}
The solutions offered by Bayesian statistics, while appealing and sufficient, could not be implemented due to computational constraints in the past. However, these limitations have been reduced with the evolution of programming languages, the development of efficient algorithms, the capacity of computers to store large amounts of data, and their ability to handle  more complex models, as explained \cite{andri}.

One of the most widely used algorithms used to address Bayesian statistics models is HMC, originally known as Hybrid Monte Carlo and introduced in \cite{duane1987hybrid} in the context of lattice field theory, before being further developed in the statistical literature. The term ``Hamiltonian Monte Carlo'' was first introduced in \cite{mac} where the method is discussed in a statistical context. A detailed treatment of HMC is provided in \cite{liu2001monte}, where the algorithm is presented within the broader framework of Monte Carlo methods, highlighting its computational advantages. In addition, \cite{neal2011mcmc} offers a comprehensive exposition of HMC, including both theoretical insights and practical implementation strategies, and remains a standard reference in the literature.

More recently, the understanding of HMC has been significantly refined. The work of \cite{betancourt2017conceptual} provides a conceptual introduction to HMC, developing geometric intuition and offering insight into why the method performs well, as well as the conditions under which it succeeds or fails, while \cite{bou2018geometric} establishes a rigorous connection between HMC and geometric numerical integration. These developments have led to further methodological advances, including adaptive and randomized variants that aim to improve efficiency and robustness in high-dimensional settings~\cite{hoffman2014no,vishwanath2024repelling}.

 To mention  a few examples, in recent years HMC has been used in credit risk to predict potential loan defaults \citep{credit_risk}; in astrophysics to detect low-frequency gravitational waves using pulsar timing arrays (PTAs) \citep{gravitational_wave}; in structural engineering, to address the inverse problem of damage identification in structures \citep{structural}; in seismology, to model time-varying seismicity rates using deep Gaussian processes \citep{muir2023deep}; in traffic safety, to investigate factors influencing the severity of injuries to car drivers, car passengers, and truck occupants in car-truck collisions \cite{traffic}; in forensic genetics, to address the challenges of analyzing mixed DNA profiles \cite{forensics}; in statistics, to infer and predict nonparametric probability density functions (PDFs) using constrained Gaussian processes \cite{staa}; in epidemiology, to develop a simulation model to understand the spread of the COVID-19 virus immediately after an infected person coughs or sneezes \cite{covid}; and in medical imaging, to improve the reliability of deep neural networks used for medical image segmentation \cite{medic}.

This paper introduces the theory and Python codes of HMC, and provides examples and comparisons. Section \ref{apena} introduces the fundamental principles of Hamiltonian mechanics. Section \ref{se_construc} presents the construction of the HMC algorithm, important theoretical results and its sensitivity to the choice of hyperparameters $\epsilon$ (leapfrog step size) and $L$ (trajectory length). Section~\ref{sec:modern_variants} reviews modern extensions of the algorithm, such as the No-U-Turn Sampler (NUTS) and the Dual Averaging (DA), which adaptively tune $\epsilon$ and $L$, and the repelling-attracting HMC (raHMC), which helps the chain visit all local modes. Section \ref{compacompa} presents illustrative examples and comparative analyses between HMC (and its variants) and RWMH in terms of execution time, convergence, effective sample size (ESS) and seconds per ESS. The Python code for these examples can be accessed via the GitHub repository at  \url{https://github.com/unabrahams/Hamiltonian-Monte-Carlo-}. Finally, Section \ref{conclu} presents the conclusions of this work.

\section{An introduction to Hamiltonian Mechanics}\label{apena}

This section introduces the fundamental elements of Hamiltonian mechanics, which are essential for comprehending the underlying principles of the HMC algorithm. This content is presented in a concise manner and is organized in a way that does not presuppose any prior knowledge of mechanics.

\subsection{Analytical Mechanics}

Analytical mechanics provides an abstract formulation of Newtonian mechanics. In contrast to the approach taken in classical mechanics, where the equations of motion are derived from vector functions and the positions of particles and the forces acting on them are defined in three-dimensional space, the equations of motion in this approach are derived from scalar functions. Moreover, their mathematical formulation is independent of any change of coordinates, thus facilitating their application to a variety of problems \citep{garrigos1}.

\subsubsection{Basic concepts}

In the context of physics, the degrees of freedom of a system are defined as the minimum number of independent scalar values that are required to determine the position of the particles within a given space. In Newtonian mechanics, the position of $n$ particles moving without constraints can be determined by using $n$ position vectors, resulting in a total of $3n$ degrees of freedom. This is discussed in greater detail in the 
\hyperref[appA]{Appendix}.

The generalized coordinates are a set of parameters $q_{1},q_{2},...,q_{s}$ that allow us to determine in an uniquely way the position of the particles in space, generally these parameters are typically measurable physical quantities, including positions along coordinate axes, angles, and distances.

In a particle system described by $s\in\mathbb{Z}_+$ generalized coordinates, the function $q : \mathbb{R} \to \mathbb{R}^s$ represents the generalized coordinates of the system at time $t$. Specifically,  $q(t) = (q_1(t), q_2(t), \dots, q_s(t))$, where $q_i(t)$ denotes the position along the $i$-th coordinate at time $t$. The objective of employing generalized coordinates is to derive the equations of motion for a system in a more efficient way. In particular, we seek a function  $\phi _{i}:\mathbb{R}^{s+1}\rightarrow{\mathbb{R}}^{3}$ such that
\begin{equation}
r_{i}\left(t\right)=\phi_{i}\left(q\left(t\right),t\right),
\end{equation}
where $r_{i}\left(t\right) \in \mathbb{R}^{3}$ represents the position of the $i$-th particle in cartesian coordinates, thus the complete position of the system in three-dimensional space can also be known, given the value of $q\left(t\right)$.

We define generalized velocities at time $t$ as the vector of derivatives of the generalized coordinates, denoted by  $\frac{dq\left(t\right)}{dt}$. We will refer to the phase space of velocities, which we denote by $M$, as the subset of $\mathbb{R}^{2s}$ such that for all times $t$,

\begin{equation}
\left(q_{1}\left(t\right),q_{2}\left(t\right),...,q_{s}\left(t\right),\frac{dq_{1}\left(t\right)}{dt},\frac{dq_{2}\left(t\right)}{dt},...,\frac{dq_{s}\left(t\right)}{dt}\right)\in M.
\end{equation}

\subsubsection{The Principle of Least Action and Euler--Lagrange Equations}\label{capi}

The reformulation of Newtonian mechanics is founded upon the principle of least action, which is also known as Hamilton's principle. In order to gain an understanding of this principle, it is first necessary to introduce the Lagrangian and the action.

The Lagrangian of a particle system is a function $L: \mathbb{R}^{2s+1} \rightarrow \mathbb{R}$, dependent on time, coordinates, and generalized velocities. It is not unique; as discussed in \citep{gera}, there are infinitely many functions from which the equations of motion can be derived. Generally, one works with the standard Lagrangian, which depends on time only through the generalized coordinates and velocities and is given by the difference between the kinetic and potential energy, both of which are defined in the \hyperref[appA]{Appendix},

\begin{equation}
\label{lag}
L\left(q\left(t\right), \frac{d q\left(t\right)}{dt}  \right)=T\left(\frac{d q\left(t\right)}{dt} \right)-U\left(q\left(t\right)\right).
\end{equation}

A proof demonstrating why the Lagrangian can be expressed in this form, beginning with D'Alembert's principle, which states that the external and inertial forces acting on a body are in equilibrium, can be found in \cite{walter}.

On the other hand, the action $A$ is a functional of the Lagrangian over a time interval $\left[t_{0}, t_{1}\right]$. It is defined as  
\begin{equation}\label{actionap}
A\left[L\right]=\sideset{}{_{t_{0}}^{t_{1}}}\int L\left(q\left(t\right),\frac{d q\left(t\right)}{dt} \right)dt.
\end{equation}

Hamilton's principle postulates that every system of particles can be described by a Lagrangian, and among all potential trajectories that the system can take within the time interval $[t_0, t_1]$, the system follows the trajectory that minimizes the action \eqref{actionap} or leads to a saddle point \citep{man}.

The trajectory followed by the system is determined by the Euler--Lagrange equations, given in \eqref{tr7}. For a particle system with $n$ degrees of freedom, these consist of $n$ second-order differential equations, where the unknowns are the generalized coordinates  $q_{1}, q_{2}, \cdots, q_{n}$, which govern the motion of the system. In \cite{man}, the derivation of the Euler--Lagrange equations is presented, illustrating the application of the principle of least action.

\begin{equation}
\label{tr7} \frac{\partial L}{\partial q_{i}}=\frac{d}{dt}\left(\frac{\partial L}{\partial \left( \frac{dq_{i}}{dt} \right) }	\right)  \text{ \hspace{0.2 cm} for all\hspace{0.2 cm}} i \in \left\{1,2,...,n\right\}.
\end{equation}

In conclusion, the Euler--Lagrange equations \eqref{tr7} describe the motion of a particle system. This approach is equivalent to Newton's laws of motion, but expresses particle motion using scalar functions rather than vector functions.

\subsubsection{Hamilton's equations and properties}
\label{sec_ecuaham}

The Hamiltonian approach offers an alternative methodology for describing motion in particle systems, and serves as the foundation for the development of quantum mechanics \citep{ponce}. Hamilton's equations form a system of $2n$ first-order differential equations, which are equivalent to the $n$ second-order Euler--Lagrange differential equations presented in  \eqref{tr7}. These equations are obtained through the application of a Legendre transformation to the Lagrangian $L$. The concept of this transformation is presented in the Definition \eqref{tl}.

\begin{definition}[Legendre transformation] 
\label{tl}
Given a function  $f\left(x_{1},x_{2},...,x_{n}\right)$  whose partial derivatives $\frac{\partial f}{\partial x_{i}}$ exist and are non-zero, we define the Legendre transformation of $f$ as the function
\begin{equation}
g\left(\textbf{s},\textbf{x}\right)=\underset{i \in S}{\sum}s_{i}x_{i}-f(x_{1},x_{2},...,x_{n}),
\end{equation}
where $S\subset \{1, 2, \dots, n\}$, $s_i = \frac{\partial f}{\partial x_i}$, 
$\textbf{s} = \mbox{$\{s_i : i \in S\}$}$, and \mbox{$\textbf{x} = \{x_i : i \notin S\}$}.

\end{definition} 

An important property of this transformation is that we can recover the original function $f$ from its Legendre transformation $g$ by applying to $g$ a Legendre transformation changing the variables $s_{i}$ attached to $g$.

Applying a Legendre transformation to the Lagrangian \eqref{lag} with respect to the $n$ generalized velocities yields the Hamiltonian $H$, defined by
\begin{equation}
\label{hamilt}H\left(q\left(t\right),p\left(t\right),t\right)=\stackrel[i=1]{n}{\sum}p_{i}\left(t\right)\frac{dq_{i}\left(t\right)}{dt}-L\left(q\left(t\right),\frac{dq\left(t\right)}{dt},t\right),
\end{equation}
where
\begin{equation}
\label{mome}
p_{i}\left(t\right)=\frac{\partial L\left(q\left(t\right),\frac{dq_{i}\left(t\right)}{dt},t\right)}{\partial \left(  \frac{dq_{i}\left(t\right) }{dt} \right)} \text{ \hspace{0.2 cm} for \hspace{0.2 cm}} i \in \left\{1,2,...,n\right\},
\end{equation}
as shown in \cite{man}.

The functions $p_{i}\left(t\right)$  presented in \eqref{mome}, are known as the generalized momenta. We will denote by  $p\left(t\right)$ the vector $\left(p_{1}\left(t\right), p_{2}\left(t\right), ..., p_{n}\left(t\right)\right)$ and define the phase space, denoted by $Q$ with $Q\subseteq\mathbb{R}^{2n}$, as the space to which the generalized coordinates and momenta belong.

By applying the Legendre transformation to the Lagrangian, we are able to map from the phase space of velocities to phase space. In this new space, the motion of the generalized coordinates and momentum can be described by first-order differential equations. Furthermore, the Legendre transformation can be reapplied in order to retrieve the generalized velocities, thus enabling the determination of trajectories within the phase space of velocities.

In order to make the notation of the Hamiltonian \eqref{hamilt} more concise, we will omit writing the time dependence of the generalized coordinates, velocities, and momenta, that is:
\begin{equation}
\label{hamilt1}H\left(q,p,t\right)=\stackrel[i=1]{n}{\sum}p_{i}\frac{dq_{i}}{dt}-L\left(q,\frac{dq}{dt},t\right).
\end{equation}

Hamilton's equations describe the trajectories followed by the generalized coordinates and momenta associated with particle systems in phase space. They are derived by calculating infinitesimal displacement or change in the Hamiltonian as the $j$-th generalized coordinate and the $j$-th generalized momentum
\begin{equation}\label{hamiltonianeq}
\frac{\partial H(q,p,t)}{\partial q_{j}} = -\frac{dp_{j}}{dt}
\quad \text{and} \quad
\frac{\partial H(q,p,t)}{\partial p_{j}} = \frac{dq_{j}}{dt},
\quad \text{for all} \enskip j \in \{1,\ldots,n\}.
\end{equation}

The solution of Hamilton's equations \eqref{hamiltonianeq} determines the trajectory of the generalized coordinates and momenta associated with a system of particles. In the case where the Lagrangian defined in \eqref{lag} is used, the Hamiltonian \eqref{hamilt1}  is given by 
\begin{equation}
\label{hamest}
H\left(q,p\right)= T\left(\frac{d q}{dt} \right)+ U\left(q\right), 
\quad\text{where}\quad
p=m\frac{dq}{dt}
\end{equation}

It is of interest to note that the Hamiltonian \eqref{hamest} is the sum of the kinetic and potential energy. Thus, the Hamiltonian \eqref{hamest} is equal to the mechanical energy defined in the \hyperref[appA]{Appendix}. 

Writing the kinetic energy of the Hamiltonian \eqref{hamest} as a function of the generalized momentum will be useful in Section \ref{se_construc}, where the construction of the HMC method is presented. If $K$ is defined as the kinetic energy in terms of the generalized momentum, $K$ will have the correspondence rule $K(p) = \frac{p^{2}}{2m}$, where $p^{2}$ is the dot product of $p$ with itself, since $K(p) = T\big(\frac{dq}{dt}\big)$. Therefore, the Hamiltonian \eqref{hamest} can be written as
\begin{eqnarray}
 H\left(q,p\right)&=& \label{hamfinfin} \frac{p^{2}}{2m}+U\left(q\right).
\end{eqnarray}

For the Hamiltonian \eqref{hamfinfin}, Hamilton's equations  \eqref{hamiltonianeq} become:
\begin{equation}
     \label{hamesteq} \frac{p_{j}}{m}=\frac{dq_{j}}{dt} \hspace{5mm} \text{and} \hspace{5mm} \frac{\partial    U\left ( q \right ) }{\partial q_{j}} = \frac{-dp_{j}}{dt} \hspace{3mm} \text{ for } j \hspace{3mm}\text{in } \left \{1,2,...,n\right \}.
\end{equation}
The HMC algorithm employs the Hamiltonian defined in \eqref{hamfinfin} as it enables the simulation of the target distribution using samples from random normal distributions, as detailed in Section \ref{se_construc}. Furthermore, Hamilton's equations possess interesting characteristics that facilitate the HMC algorithm's  functionality.

First, Hamilton's equations are reversible. In order to explain the concept of reversibility in Hamiltonian dynamics, let us define  $\delta: \mathbb{R} \rightarrow \mathbb{R}^{2n}$ as the function that at time $t$ returns the state in phase space (the space of all possible states of a physical system, in this case, the possible values of the momentum and generalized coordinates) at which a system is located, i.e.,   $\delta \left ( t \right )=\left ( q_{1}\left ( t \right ),..., q_{n}\left ( t \right ), p_{1}\left ( t \right ),..., p_{n}\left ( t \right )
 \right )$. Thus, reversibility can be defined as the ability of the system to return to the state $\delta\left(t\right)$ starting from $\delta\left(t + \varepsilon\right)$ with  $\varepsilon > 0$. 

 Reversibility in Hamiltonian dynamics, as shown in \cite{hand} and stated in Theorem \ref{teorever}, is achieved by following the trajectory generated by Hamilton's equations for a time $\varepsilon$ while reversing the sign of the generalized momentum.

  \begin{theorem}\label{teorever}
The Hamiltonian dynamics associated with the Hamiltonian $H(q,p)$ are reversible under the transformation  $H(q,p) \rightarrow H(q,-p)$.
 \end{theorem}
 \begin{proof}  

By changing the sign of $p$, the Hamiltonian equations become
\begin{eqnarray}
\frac{p_{j}}{m}=-\frac{dq_{j}}{dt} 
\quad \text{and} \quad 
\frac{\partial U(q) }{\partial q_{j}} = \frac{dp_{j}}{dt} 
\quad \text{for} \quad
j \in \left \{1,2,...,n\right \}.
\end{eqnarray}
Since $ -\frac{dq_{j}}{dt} = \frac{dq_{j}}{d(-t)}  $   and   $\frac{dp_{j}}{dt} = -\frac{dp_{j}}{d(-t)}$, it follows that 
\begin{eqnarray}
\frac{p_{j}}{m}=\frac{dq_{j}}{d(-t)} \hspace{5mm} \text{and} \hspace{5mm} \frac{\partial    U\left ( q \right ) }{\partial q_{j}} = -\frac{dp_{j}}{d(-t)}, \hspace{3mm} 
 \end{eqnarray}
which proves that by considering  $H(q,-p)$ instead of  $H(q,p)$, 
one obtains the time-reversed Hamiltonian dynamics, i.e., with the transformation $t\rightarrow -t$.
\end{proof}

Obtaining the inverse dynamics by changing the sign of the momenta in  Hamilton's equations is helpful in the HMC method because it makes it easier to obtain the inverse trajectories in the algorithm. Attaining reversibility is crucial because it enables the Markov chain generated by the HMC algorithm to satisfy the detailed balance equations presented in Section \ref{sec_ebd} and also permits the introduction of new states for the Markov chain that do not cancel the acceptance and rejection rate, as explained in Section \ref{se_construc}.

Another important property of HMC is that the Hamiltonian dynamics preserve the Hamiltonian $H(q, p)$ constant, as stated next and proved in \cite{hand}.

\begin{theorem} \label{teo_pfel}
The value of the Hamiltonian  $H\left(q,p\right)$ remains invariant on the trajectories of the Hamiltonian dynamics. 
\end{theorem}

This result implies that the acceptance probability in the HMC is one and that the points on the trajectories of the Hamiltonian dynamics have the same associated density according to the Boltzmann distribution, as mentioned in the Section \ref{se_construc}. Further, Hamiltonian dynamics are volume-preserving in phase space. This is known as Liouville's theorem, presented next and proved in \cite{hand}.

\begin{theorem}[Liouville's theorem] \label{teo_lio}
The volume of a region $\omega$ in the phase space is conserved if the points on its boundary $d\left(\omega\right)$ move according to Hamilton's equations.
\end{theorem}

The importance of volume preservation for the HMC algorithm lies in the fact that if it is not preserved, one would have to compute the determinant of the Jacobian matrix of the transformation each time a new state is proposed, in order to modify the acceptance rate, which would be computationally expensive~\cite{hand}. 

Finally, it is important to mention the  symplectic property of Hamiltonian dynamics. To define this property, it is essential to consider the definitions presented in \eqref{mat_sim} and \eqref{tra_sim}.

\begin{definition}[Symplectic Matrix] \label{mat_sim}
We say that a matrix $M$ is symplectic if  
\[
M^{t}JM=J, 
\quad\text{where}\quad J=\begin{pmatrix}0_{n\times n} & I_{n\times n}\\
-I_{n\times n} & 0_{n\times n}
\end{pmatrix}.
\]
\end{definition} 

\begin{definition}[Symplectic Transformation]\label{tra_sim}
A transformation $\varphi$ is symplectic if its Jacobian matrix $M$ is symplectic.
\end{definition} 

The following theorem, due to Henri Poincar\'e, establishes that Hamiltonian dynamics are symplectic if the Hamiltonian $H\left(q,p\right)$ has continuous partial derivatives of order two. Its proof can be found in \cite{simp}.

\begin{theorem}\label{ponca}
If $H\left(q,p\right)$  is a Hamiltonian with continuous second partial derivatives, then the transformation implied by the Hamilton's equations associated with  $H\left(q,p\right)$ at time $t$, is a symplectic transformation. 
\end{theorem}

The volume-preserving property of Hamiltonian dynamics is a consequence of their symplectic nature. This property arises because the flow generated by Hamilton's equations is symplectic, as ensured by Theorem \ref{ponca}. Such transformations are volume-preserving. Moreover, the symplectic property is important because symplectic numerical integrators provide more accurate approximations to the solutions of Hamilton’s equations, which are essential for the implementation of the HMC algorithm.

\section{Building the HMC algorithm}\label{se_construc}

The goal of the HMC method is to simulate a target distribution on $\mathbb{R}^{d}$. We assume that this target distribution has a density function $S^*(q)$, with respect to Lebesgue measure in $\mathbb{R}^{d}$. To use this algorithm, we should at least be able to evaluate $-\log(S(q))$, where $S$ is proportional to $S^*$.  Additionally, the partial derivatives of $-\log(S(q))$ with respect to $q = (q_1, q_2, \dots, q_d)$ are required to exist, and it must be possible to evaluate them. In other words, a computational method should be available to calculate both $-\log(S(q))$ and its partial derivatives for all $q \in a \subset \operatorname{supp}(S^*(q)) \subset \mathbb{R}^d$.

To adapt the \emph{Boltzmann} distribution to HMC, we need to consider a system of particles with $d$ generalized coordinates and $d$ generalized momenta, denoted by $q=\left(q_{1},q_{2},. .,q_{d}\right)$ and $p=\left(p_{1},p_{2},...,p_{d}\right)$, respectively, such that the density associated with a point  $\left(q,p\right)\in \mathbb{R}^{2d}$  is given by
\begin{equation}
\label{bolt} 
B(q,p)=\frac{1}{z}\exp\left(\frac{-E\left(q,p\right)}{c}\right).
\end{equation}
In addition, we need to use the Hamiltonian $H(q,p) = K(p)+U(q)$ introduced in Section \ref{sec_ecuaham}, where the functions $K$ and $U$ are the kinetic and potential energy, respectively, as the energy function $E$. It is also necessary to consider the temperature $\psi$ as $\frac{1}{k}$ to avoid the constant $c$. Then, \eqref{bolt} can be rewritten as
\begin{align}
\nonumber
B(q,p) 
&= \frac{1}{z} \exp(-H(q,p))
= \frac{1}{z} \exp(-K(p)) \exp(-U(q)) \\
&= \frac{1}{z} \exp\left(-\sum_{i=1}^{d} \frac{p_{i}^{2}}{2m_{i}}\right) \exp(-U(q)).
\label{prop2}
\end{align}

Note that the factor in front of $\exp(-U(q))$ in \eqref{prop2} is the kernel of the multivariate normal distribution, given by 
\begin{equation}
\label{normal}
\begin{gathered}
f(x\mid \mu ,\Sigma)
=\frac{1}{\left(2\pi\right)^{\frac{n}{2}}\det\left(\Sigma\right)^{\frac{1}{2}}}\exp\left(-\frac{1}{2}\left(x-\mu\right)^{t}\Sigma^{-1}\left(x-\mu\right)\right),\\
\text{where}\quad 
x\in \mathbb{R}^{d},\enskip
\mu=\left(0,0,...,0\right)^{t}
\enskip\text{and}\enskip
\Sigma=\mathrm{diag}(m_1,\ldots,m_d).
\end{gathered}
\end{equation}

Let $w = (2\pi)^{\frac{n}{2}} \det(\Sigma)^{\frac{1}{2}}$ denote the normalization constant of the density $f$. It then follows that 
\begin{equation}
\label{400}\exp\left(-\sum_{i=1}{d}\frac{p_{i}^{2}}{2m_{i}}\right) 
=w  f\left(p\right).
\end{equation}

From  \eqref{400} and    \eqref{prop2} it follows that 
 \begin{eqnarray}
 B\left(q,p\right)&=& \frac{1}{z}\exp\left(-\stackrel[i=1]{d}{\sum}\frac{p_{i}^{2}}{2m_{i}}\right)\exp\left(-U\left(q	\right)\right)\nonumber\\
 &=&\label{raaa}  \frac{w}{z} f\left(p\right)  \exp\left(-U\left(q	\right)\right).
\end{eqnarray}
In addition, we must consider to the potential energy $U$  as 
\begin{equation}
\label{ju} U\left(q\right)=-\log\left(S\left(q\right)\right),
\end{equation}
so $B\left(q,p\right)$  can be rewritten as
\begin{equation}
B(q,p)= \frac{w}{z} f\left(p\right)  \exp\left(\log\left(S\left(q\right)\right)\right)
=\label{low}  \frac{w}{z}f\left(p\right)S\left(q\right).
\end{equation}

By constructing $B\left(q,p\right)$ in this way, its marginal density with respect to $q$ will be the target density $S'$, since
\begin{align}
\sideset{}{_{\mathbb{R}^d}}\int B\left(q,p\right)dp
&=\sideset{}{_{\mathbb{R}^d}}\int \dfrac{w}{z}f\left(p\right)S\left(q\right)dp
=\dfrac{w}{z} S\left(q\right) \sideset{}{_{\mathbb{R}^d}}\int f\left(p\right)dp\nonumber\\
&=\dfrac{w}{z} S\left(q\right)
= S'\left(q\right).
\label{resulta11}
\end{align}

Equality \eqref{resulta11} follows from the fact that $\frac{w}{z}$ is the normalization constant of the function $S$, and since $S$ is proportional to the target density $S'$, it must hold that $S'=\frac{w}{z}S$. Note also that the marginal density of $B(q,p)$ with respect to $p$ is the multivariate normal density defined in \eqref{normal}:
\begin{equation}
\sideset{}{_{\mathbb{R}^{d}}}\int B(q,p)dq
=\sideset{}{_{\mathbb{R}^{d}}}\int\frac{w}{z}f(p)S(q)dq
=f(p)\sideset{}{_{\mathbb{R}^{d}}}\int\frac{w}{z}S(q)dq
= f(p).
\label{resultt}
\end{equation}

The HMC algorithm constructs a Markov chain $\{q_{i},p_{i}:i=1,\ldots,N-1\}$ using Algorithm~\ref{algo}, whose stationary distribution has density $B(q,p)$, ensuring that the set  $\{ q_{i}:i=1,\ldots,N-1\}$ has the target distribution. The first state of the chain generated with the HMC is the point $(q^{(0)}, p^{(0)})$, where $q^{(0)}$ is a starting value that must belong to the support of $S'$, the density of the target distribution, and $p^{(0)}$ is a simulation of the normal distribution $N(\mu,\Sigma)$, whose density is given in \eqref{normal}.
 
Then we have to solve or approximate the solutions of Hamilton's equations 
\begin{equation}
    \label{hamest1}\frac{p_{j}}{m}=\frac{dq_{j}}{dt} \hspace{5mm} \text{and} \hspace{5mm} \frac{\partial U\left ( q \right ) }{\partial q_{j}} = \frac{-dp_{j}}{dt} \hspace{3mm} \text{ for } j \hspace{3mm}\in \left \{1,2,...,d\right \},
\end{equation}
at a given time $T$, considering as starting point $\left(q^{(i-1)}, p^{(i-1)}\right)$. The value $T$ is a tuning parameter. Its value depends on the target distribution to be simulated. It must be large enough so that the proposals are sufficiently far from the current state and the support of the target distribution can be efficiently explored.

The proposal for the state $(q^{(i)}, p^{(i)})$ of the chain is the point $(q^*, - p^*)$, where $(q^*, p^*)$ is the approximation or solution of Equation \eqref{hamest1}. The negative sign in the proposal is essential since, if it is not used,  the acceptance rate used in the Metropolis Hasting algorithm also required for the HMC algorithm always becomes zero. To prove this, note that according to this ratio, the proposal must be accepted with probability $\rho((q_{i-1},p_{i-1}),(q_{i},p_{i}))=\min(1,r)$, where
\begin{align}
r&=\frac{f(p_{i})S'(q_{i})}{f(p_{i-1})S'(q_{i-1})}\frac{\mathbb{P}(q_{i-1}, p_{i-1}\mid q_{i},p_{i})}{\mathbb{P}(q_{i},p_{i}\mid q_{i-1},p_{i-1})}\nonumber\\
&=\label{razon}  \frac{f(p_{i})S(q_{i})}{f(p_{i-1})S(q_{i-1})}\frac{\mathbb{P}(q_{i-1}, p_{i-1}\mid q_{i},p_{i})}{\mathbb{P}(q_{i},p_{i}\mid q_{i-1},p_{i-1})}.
\end{align}

Equality \eqref{razon} follows from the fact that the density $S\propto S'$. The ratio in \eqref{razon} cancels out because the probability of reaching the point $(q_{i-1},p_{i-1})$ at time $T$ starting from the point $(q_{i},p_{i})$ at time zero, denoted by $\mathbb{P}(q_{i-1},p_{i-1}\mid q_{i}, p_{i})$, is zero. Indeed, this is the case because, to return to the same state, we have to get a certain $p$ value, if it exists, from the normal distribution $N(\mu,\Sigma)$. Since $p$ is a numerical quantity, its value is obtained with zero probability.

By proposing  $(q^*, - p^*)$ instead of $(q^*, p^*)$, the ratio in \eqref{razon} does not cancel out since, as explained in Theorem \ref{teorever}, the reversibility of the Hamiltonian dynamics implied by the  Hamiltonian is achieved by changing the generalized momentum $p$ for $-p$. Thus, the chain will deterministically reach the state $(q_{i-1},p_{i-1})$ starting from $(q_{i},-p_{i})$, i.e., $\mathbb{P}(q_{i-1},p_{i-1}\mid q_{i},-p_{i})=1$. Moreover, since $\mathbb{P}(q_{i},-p_{i}\mid q_{i-1},p_{i-1})=1$, because this deterministic transition arises from Hamilton's equations, the ratio in \eqref{razon} reduces to:
\begin{align}
r&=\frac{f(-p_{i})S(q_{i})}{f(p_{i-1})S(q_{i-1})}\frac{\mathbb{P}(q_{i-1}, p_{i-1}\mid q_{i},-p_{i})}{\mathbb{P}(q_{i},-p_{i}\mid q_{i-1},p_{i-1})}\nonumber\\
 &=\frac{f(-p_{i})S(q_{i})}{f(p_{i-1})S(q_{i-1})}
=\frac{f(p_{i})S(q_{i})}{f(p_{i-1})S(q_{i-1})}.
\label{razmas2}
\end{align} 
The last equality in \eqref{razmas2} follows from the fact that $f$ is a symmetric density at zero, and hence $f(-p_{i})=f(p_{i})$.

We will simplify the ratio in \eqref{razmas2} in order to reduce the computational cost associated with the operations involved in calculating the acceptance and rejection rates in each iteration. From Equation \eqref{prop2}, $B(q,p)= \frac{1}{z} \exp(-H(q,p))$, and by \eqref{low}, we have that $B(q,p)= \frac{w}{z} f(p)S(q)$. Thus, it follows that
\begin{align}
\nonumber
r&=\frac{f(p_{i})S(q_{i})}{f(p_{i-1})S(q_{i-1})}
=\frac{\frac{w}{z}(p_{i})S(q_{i})}{\frac{w}{z}f(p_{i-1})S(q_{i-1})}\\
&=\frac{\exp(-H(q_{i},p_{i}))}{\exp(-H(q_{i-1},p_{i-1}))}
=\label{razon12} \exp(H(q_{i-1},p_{i-1})-H(q_{i},p_{i})).
\end{align}

Ultimately, a value $u$ is generated from a uniform distribution $U(0,1)$. 
In the event that $u<min(1,r)$, the point $(q_{i},p_{i})$ is accepted as the subsequent state of the chain. Otherwise, the current state is assigned as the next state of the chain. This procedure is summarized in Algorithm \ref{algo}.

In this work, we will use the Leapfrog algorithm presented in \cite{nume} to approximate the solution of the system of equations \eqref{hamest1}. The Leapfrog algorithm is commonly used in HMC because it shares the symplectic property of Hamiltonian dynamics, as introduced in Section \ref{sec_ecuaham}. As demonstrated in reference \cite{nume}, the Leapfrog method is proven to be a symplectic method when applied to the solution of Hamilton's equations of the form \eqref{hamest1}. The advantage of utilising the Leapfrog algorithm is that it preserves volume and, due to its symmetry, is also reversible. Consequently, the Leapfrog algorithm exhibits a preferable global and local error order in comparison to alternative methods, as shown in Table \ref{tabla_errores} consulted in \cite{hand}.

\begin{table}
\caption{\label{tabla_errores}Error order of numerical methods.}
\begin{tabular}{@{}lccccc@{}}
\hline 
 & Euler & Symplectic Euler  & Leapfrog \\ 
 \hline 
Local error   & $\epsilon^{2}$ &  $\epsilon^{2}$ & $\epsilon^{3}$\\  
Global error  & $\epsilon$  & $\epsilon$ & $\epsilon^{2}$ \\ 
 \hline
\end{tabular}
\end{table}

In the context of the Leapfrog algorithm, the solution of the system of equations \eqref{hamest1} is approximated by defining $T$ as $L\epsilon$, where $L$ represents the number of steps of length $ \epsilon$ taken by the Leapfrog algorithm. The parameter $\epsilon$  serves to regulate the degree of precision with which the leapfrog algorithm approximates the trajectory of the Hamiltonian dynamics. As detailed in Section \ref{tasahh}, an appropriate approximation of the solution to the Hamiltonian equation is crucial for achieving a high acceptance rate in the HMC algorithm.

The value of the parameter $L$ must be selected based on the chosen value of the parameter $ \epsilon$. If a small value of $ \epsilon$ is selected, the value of $L$ should be large enough to ensure that the length of the path traversed by the chain in a single iteration, represented by the value $L\epsilon$, is sufficient to efficiently explore the target distribution and avoid autocorrelation in the chain. Choosing optimal values for $\epsilon$ and $L$ is a delicate balancing act: a step size too large results in low acceptance rates due to discretization errors, while a step size too small leads to inefficient, random-walk-like behavior.

The challenge in selecting the tuning parameters $L$ and $ \epsilon$ lies in the fact that, in general, the optimal parameters vary before and after convergence to the target distribution. Furthermore, it is possible that the optimal tuning parameters may differ for different regions of the target distribution. For the basic HMC, \cite{hand} proposes choosing $L$ through trial and error. As a starting point, the value of $L=100$ is suggested. If, after simulating values from the target distribution, non-significant autocorrelations are achieved, a smaller $L$ can be considered; otherwise, a larger one should be chosen. On the other hand, it is recommended that $\epsilon$ be chosen randomly within a small interval.

While these position-dependent sensitivities present a substantial bottleneck for standard HMC, modern algorithmic extensions were developed explicitly to resolve these limitations. As detailed in Section \ref{sec:modern_variants}, the Dual Averaging scheme adaptively optimizes the step size $\epsilon$ to achieve a target acceptance rate, eliminating the need for manual step-size tuning. Furthermore, algorithms such as NUTS dynamically determine the trajectory length $L$ at each individual iteration based on the local geometry of the target density. This allows the algorithm to take short trajectories in regions of high curvature and long trajectories in flat regions, effectively automating the hyperparameter selection that standard HMC struggles with.

\begin{algorithm}
\caption{Hamiltonian Monte Carlo}
\label{algo}
\begin{algorithmic}
\STATE\textbf{Input:} potential $U(q)=-\log(S(q))$, where $S$ is proportional to the target density  $S'$, initial point $q[0]=q^{(0)}$, length of the trajectory $T$, chain size $N$.
\FOR{$i = 1,\ldots, N-1$}
\STATE Sample $p^{(i-1)}\sim N(0,\Sigma)$.
\STATE Let $(q^{*},p^{*})$ be the solution (or an approximation) to Hamilton's equations 
\begin{equation*}
\frac{\partial H(q,p)}{\partial p_{j}}=\frac{-dq_{j}}{dt}, \quad 
\frac{\partial H(q,p)}{\partial q_{j}}=\frac{dp_{j}}{dt}, \quad j\in\{1,\ldots,\mathrm{length}(q)\},
\end{equation*}
with initial point $(q^{(i-1)},p^{(i-1)})$ at a time $T$.

\STATE Sample $u\sim U\left(0,1\right)$

\IF{$u\leq \min(1,\exp(H(q^{(i-1)},p^{(i-1)})-H(q^{*},-p^{*})))$}
\STATE Set $q[i]=q^{*}$ 
\ELSE
\STATE Set $q[i]=q^{(i-1)}$ 
\ENDIF
\ENDFOR
\STATE \textbf{return} $\{ q[0], q[1],\ldots,  q[N-1]\}$ 
\end{algorithmic}
\end{algorithm}

\subsection{Theoretical results}

This section presents theoretical results regarding the acceptance rate and the theory of Markov chains that support the HMC algorithm.

\subsubsection{Acceptance probability} \label{tasahh}

The proposals in the HMC method are theoretically accepted with probability one.  To support this assertion, one must recall that in Section \ref{se_construc}, it was demonstrated that the probability of acceptance is given by the $\min(1,r)$, where 
\begin{equation}
r = \exp(H(q_{i-1},p_{i-1})-H(q_{i},p_{i})).
\end{equation}

Additionally, according to Section \ref{sec_ecuaham}, Hamiltonian dynamics keep the Hamiltonian invariant, so $H(q_{i}, p_{i}) = H(q_{i+1}, p_{i+1})$, and therefore
\begin{equation}
\label{unoh}
r=\exp(H(q_{i-1},p_{i-1})-H(q_{i},p_{i}))
=\exp(0)
=1.
\end{equation}
Hence, the probability of accepting a new state for the chain is $\min(1, r) = 1$. 

It is essential to note that, in practice, proposals are not accepted with probability one. This is because the Hamiltonian does not remain constant as a result of the inherent inaccuracies associated with the numerical methods employed to approximate the Hamiltonian dynamics (i.e., to solve Hamilton's equations). In fact, the higher the accuracy of the approximation method used, the closer $H(q_{i-1}, p_{i-1}) - H(q_{i}, p_{i})$ is to zero. Consequently, the value of $r$ in \eqref{unoh} will be closer to one. Therefore, the greater the precision of the numerical method employed, the closer the acceptance rate will be to one.

Although, theoretically, all proposals are accepted, this is suboptimal. In \cite{nota2} it was shown that when the dimension of the target distribution tends to infinity (and under other assumptions imposed on the Leapfrog algorithm), the optimal acceptance rate in HMC is $65.1 \%$. 

\subsubsection{Detailed balance equations and stationary distribution}\label{sec_ebd}

One of the conditions that allows the simulation of a target distribution with MCMC methods is that the chain generated exhibits a stationary distribution that is equal to the target distribution.

In \cite{hand}, it is shown that the function $B(q,p)$, as defined in Equation \eqref{prop2}, satisfies the detailed balance equations, defined in \cite{advanced}, with respect to the Markov chain $\{ (q_{i},p_{i}):i=0,1,\ldots,N-1\}$ generated by the HMC algorithm. This implies that the stationary distribution of the chain has a density function $B(q,p)$, as supported by the next result (see \cite{advanced} for a proof).

\begin{theorem} \label{teo_imp}
Let $\{X_{n} : n = 0, \dots, N-1\}$ be a Markov chain on $\mathbb{X}$ with a transition kernel $P(dy \mid x)$ and $\pi(dy)$ a probability measure on $(\mathbb{X}, \mathcal{B})$. Let $f$ denote the density function of the distribution $\pi(dy)$, and let $Q(y \mid x)$ represent the conditional density function of the kernel $P(dy \mid x)$. If $f$ satisfies the detailed balance equations, then $\pi(dy)$ is the stationary distribution of $\{X_{n} : n = 0, \dots, N-1\}$.
\end{theorem}

\subsubsection{Ergodicity of the chain}\label{sec_ergodica}

In Section \ref{sec_ebd} we discussed the stationary distribution of the Markov chain generated by the HMC algorithm, which has a density $B(q,p)$. However, this result is insufficient to guarantee that the simulated chain will converge to the target distribution. To ensure this convergence, the chain must satisfy the hypothesis of the ergodic theorem stated in \cite{kass} and recalled next.

\begin{theorem}[Ergodic theorem] \label{teo_erg}
Let $\{ X_{n}:n=0,1,2,...,N-1\}$ be an aperiodic and positive Harris recurrent Markov chain with state space $\mathbb{X}$, stationary distribution  $\pi(dx)$  on $(\mathbb{X} ,\mathcal{B} )$ and $g$ a real function.  If $\mathbb{E}_{\pi}(|g(x)|) < \infty  $, then 
\begin{equation}
\hat{\mu}_{g}\overset{a.s}{\rightarrow}\mathbb{E}_{\pi}(g(X)),
\end{equation}
where $\hat{\mu}_{g}=\frac{1}{N}\sum_{i=0}^{N-1}g(X_{i})$ and  $\mathbb{E}_{\pi}(g(X))=\int_{\mathbb{X}}g(x) \pi(dx)$.
\end{theorem}

A stronger property sought in MCMC methods is that the generated chain be geometrically ergodic \cite{tierney} (see Definition \ref{geometrica}, below). The importance of this property, as stated in \cite{ergodicoMH}, is that it implies that the central limit theorem proved in \cite{kass}, is satisfied for the chain's simulated states. Moreover, it implies that the convergence to the stationary distribution occurs at an exponential rate.

\begin{definition}[Geometrically ergodic Markov chain]\label{geometrica}
Consider a Markov chain $\{X_n : n = 0, 1, 2, \dots, N-1\}$ with state space $\mathbb{X}$, stationary distribution $\pi(dy)$, and $n$-step transition kernel $P^n(dy \mid x)$. The chain is said to be geometrically ergodic if there exists a constant $\lambda\in[0,1)$ such that  
\begin{equation}\label{eqgeom}
\| P^n(\cdot\mid x) - \pi(\cdot) \| 
\leq M(x) \lambda^n,
\end{equation}
where $M(x)$ is a real-valued, integrable function.
\end{definition}

 In general, it cannot be guaranteed that the Markov chain generated by the HMC algorithm is geometrically ergodic; however, there are some important articles that discuss in detail the conditions under which this property is satisfied. 

In \cite{ergodicidad}, it is proven that under independent position integration times, the conditions for geometric ergodicity are essentially a gradient of the log-density which asymptotically points towards the center of the space and grows no faster than linearly. For the case of position dependent integration times, it is shown that a much broader class of tail behaviors can generate geometrically ergodic chains. This paper also presents some cases in which geometric ergodicity cannot be guaranteed.

Another important text that discusses in depth the geometric ergodicity property in HMC  is \cite{ergo2}. In that work, it is proved under which conditions in the Leapfrog algorithm and target distribution, the chain is irreducible and positive recurrent. Then, under more stringent conditions, it is shown that the chain generated is Harris recurrent and geometrically ergodic. 

The present work does not cover all the conditions under which HMC generates a geometrically ergodic Markov chain, as they are extensive and fall outside the main focus of this work. Readers interested in learning more about this topic are encouraged to consult \cite{ergodicidad} and \cite{ergo2}.

\section{Modern variants of Hamiltonian Monte Carlo}
\label{sec:modern_variants}

In this section, we review some modern variants of the HMC algorithm designed to overcome the inherent sensitivities of standard Hamiltonian dynamics. As demonstrated in previous sections, the efficiency of standard HMC is heavily bottlenecked by the need to manually tune the leapfrog step size $\epsilon$ and the trajectory length $L$. Furthermore, standard HMC often struggles to traverse the low-probability energy barriers present in multimodal distributions. 

To address these limitations, we first discuss the No-U-Turn Sampler (NUTS) \cite{hoffman2014no}, which dynamically determines appropriate the trajectory length $L$ to prevent the Markov chain from retracing its steps (``U-turning'') and wasting computational effort. Next, we review the unadjusted HMC with Stratified Monte Carlo time integration (uHMC with SMC) \cite{bou2025unadjusted}. Finally, we explore the repelling-attracting HMC (raHMC) \cite{vishwanath2024repelling}, a method that modifies the system's momentum using friction to efficiently escape local modes. When combined with the Dual Averaging scheme introduced in Section \ref{subsec:DA} to adapt $\epsilon$, these variants resolve the hyperparameter tuning bottleneck. We compare their empirical performance in Section \ref{compacompa}.

\subsection{The No-U-Turn Sampler}
As it was already discussed in Section \ref{se_construc}, it is well documented in the literature that the performance of the HMC algorithm is sensitive to the choice of $\epsilon$ and $L$; in particular, some specifications may lead to poor performance of the algorithm.

Consequently, the effective use of HMC may require considerable expertise and experience in tuning these parameters, or alternatively, performing several preliminary runs of the algorithm in order to find an appropriate parametrization. The NUTS algorithm aims to use HMC without the need to tune the parameter $L$. To achieve this, it introduces an adaptive mechanism based on certain geometric criteria that we briefly describe below.

The main criterion is based on the dot product between the current momentum and the difference between the current position and the initial position of the trajectory. The idea is to continue simulating the leapfrog dynamics until this quantity becomes negative. Intuitively, this criterion detects the moment when the trajectory begins to return toward the initial position; that is, when the proposals start to ``turn back''. Intuitively, this wastes computational resources by proposing a state highly correlated with the initial state despite the computational effort. The NUTS monitors the distance between the proposed state and the initial state, terminating the leapfrog integration the exact moment the trajectory begins to double back on itself.

However, this criterion alone is not sufficient. In order to guarantee reversibility of the algorithm, at each iteration a trajectory is constructed recursively whose number of leapfrog steps doubles at each expansion, while the direction of exploration (forward or backward in time) is chosen at random. This constructs a tree whose nodes correspond to possible candidate states for the chain. For further details on this procedure, we refer the reader to \cite{hoffman2014no}.

Next, we present the pseudocode of an efficient version of the NUTS algorithm, using a transition kernel that leaves invariant the uniform distribution over the set of candidate states, from which the proposed state is ultimately selected. Further, we can adaptively tune $\epsilon$ using the Dual Averaging scheme. Since NUTS does not involve a single Metropolis accept-reject step, an alternative statistic to the acceptance probability must be defined. Set
\[
H_t^{\text{NUTS}} := \frac{1}{|B_t^{\text{final}}|} \sum_{(\theta,r)\in B_t^{\text{final}}} 
\min\left(1, \frac{p(\theta,r)}{p(\theta_{t-1}, r_{t,0})} \right),
\quad\text{for each}\enskip
t\in\mathbb{Z}_+,
\]
and define its expectation in stationarity as
\[
h^{\text{NUTS}} := \mathbb{E}_\pi\big[ H_t^{\text{NUTS}} \big],
\]
where $B_t^{\text{final}}$ denotes the set of states explored during the final doubling stage of iteration $t$, and $(\theta_{t-1}, r_{t,0})$ is the initial state and resampled momentum.

The quantity $H_t^{\text{NUTS}}$ can be interpreted as the average acceptance probability that a standard HMC step would assign to the states visited in the final doubling phase. Assuming that $h^{\text{NUTS}}$ is nonincreasing in $\epsilon$, one can apply the Dual Averaging scheme with
\[
H_t \equiv \delta - H_t^{\text{NUTS}}, \qquad x = \log \epsilon,
\]
to adapt $\epsilon$ so as to enforce $h^{\text{NUTS}} \approx \delta$ for a target $\delta \in (0,1)$.

\begin{algorithm}[ht]
\caption{{Efficient No-U-Turn Sampler}~\cite{hoffman2014no}}
\begin{algorithmic}
\STATE \textbf{Input:} initial $q_0$, step size $\epsilon>0$, potential $U=-\log S$, sample size $N\in\mathbb{Z}_+$

\FOR{$m = 1$ to $N$}

\STATE Sample $p_0 \sim \mathcal{N}(0,I)$ and $u \sim \mathrm{Uniform}\left([0,\exp\{-U(q_{m-1})-\tfrac12 p_0^\top p_0\}]\right)$

\STATE Initialize $q^- = q^+ = q_{m-1},\; p^-=p^+ = p_0,\; j = 0,\; q_{m} = q_{m-1},\; n = 1,\; s = 1$

\WHILE{$s = 1$}

\STATE Choose a direction $v_j \sim \mathrm{Uniform}(\{-1,1\})$

\IF{$v_j = -1$}

\STATE $q^-, p^-, -, -, q', n', s' \leftarrow \text{BuildTree}(q^-,p^-,u,v_j,j,\epsilon)$

\ELSE

\STATE $-, -, q^+, p^+, q', n', s' \leftarrow \text{BuildTree}(q^+,p^+,u,v_j,j,\epsilon)$

\ENDIF

\IF{$s' = 1$}

\STATE With probability $\min\{1,\frac{n'}{n}\}$ set $q^{m} \leftarrow q'$

\ENDIF

\STATE Update $n \leftarrow n + n'$, $j\leftarrow j+1$

\STATE Update $s \leftarrow s' \mathbf{1}[(q^+ - q^-)\cdot p^- \ge 0]\mathbf{1}[(q^+ - q^-)\cdot p^+ \ge 0]$

\ENDWHILE

\ENDFOR

\STATE \textbf{return} $\{q_1,q_2,\ldots,q_N\}$

\end{algorithmic}
\end{algorithm}

\begin{algorithm}[ht]
\caption{BuildTree algorithm~\cite{hoffman2014no}}
\begin{algorithmic}

\STATE \textbf{function} BuildTree$(q,p,u,v,j,\epsilon)$

\IF{$j = 0$}

\STATE $(q',p') \leftarrow \text{Leapfrog}(q,p,v\epsilon)$

\STATE $n' \leftarrow \mathbf{1}[u \le \exp\{-U(q')-\tfrac12 p'^\top p'\}]$

\STATE $s' \leftarrow \mathbf{1}[u < \exp\{\Delta_{\max}-U(q')-\tfrac12 p'^\top p'\}]$

\STATE \textbf{return} $q',p',q',p',q',n',s'$

\ELSE

\STATE $q^-,p^-,q^+,p^+,q',n',s' \leftarrow \text{BuildTree}(q,p,u,v,j-1,\epsilon)$

\IF{$s'=1$}

\IF{$v=-1$}

\STATE $q^-,p^-,-,-,q'',n'',s'' \leftarrow \text{BuildTree}(q^-,p^-,u,v,j-1,\epsilon)$

\ELSE

\STATE $-,-,q^+,p^+,q'',n'',s'' \leftarrow \text{BuildTree}(q^+,p^+,u,v,j-1,\epsilon)$

\ENDIF

\STATE With probability $\frac{n''}{n'+n''}$ set $q' \leftarrow q''$

\STATE Update $n' \leftarrow n' + n''$, $s' \leftarrow s''\mathbf{1}[(q^+ - q^-)\cdot p^- \ge 0]\mathbf{1}[(q^+ - q^-)\cdot p^+ \ge 0]$

\ENDIF

\STATE \textbf{return} $q^-,p^-,q^+,p^+,q',n',s'$

\ENDIF

\end{algorithmic}
\end{algorithm}

\subsection{Dual Averaging}
\label{subsec:DA}
We can avoid the necessity of tuning the parameter $\epsilon$ by using a stochastic approximation scheme. The most basic is that introduced by \cite{robbins1951stochastic}. However, in MCMC applications, the Dual Averaging (DA) scheme is preferred (cf.~\cite{nesterov2009primal}). The purpose of this section is to describe this scheme. Before proceeding, we note that the DA can be coupled with any of the HMC variants, thus branching every variant into two, e.g., NUTS and NUTS-DA.

Suppose that $H_t$ is an observable function whose expectation (under stationarity), conditional on a parameter $x$, is unknown and denoted by
\[
h(x)=\mathbb{E}_\pi[H_t\mid x].
\]
Our objective is to find a parameter value $x^*$ such that $h(x^*)=0$.
Under suitable assumptions on $h$, the averaged sequence $(\bar{x}_t)_{t\geq 1}$
generated by the following recursion is ensured to converge to such a value $x^*$:
\begin{equation}\label{DAscheme}
    x_{t+1} \leftarrow \mu 
- \frac{\sqrt{t}}{\gamma}\frac{1}{t+t_0} \sum_{i=1}^{t} H_i,
\qquad
\bar{x}_{t+1} \leftarrow 
\eta_t x_{t+1} + (1-\eta_t)\bar{x}_t.
\end{equation}

Here, $\mu$ is a free parameter toward which the iterates are shrunk, $\gamma>0$ controls the strength of this shrinkage, and $t_0$ helps stabilize the early iterations. The step sizes $\eta_t>0$ are chosen according to the usual Robbins--Monro decay conditions:
\[
\sum_{t=1}^{\infty} \eta_t = \infty,
\qquad
\sum_{t=1}^{\infty} \eta_t^2 < \infty.
\]3
We initialize the averaged sequence by setting $\bar{x}_1=x_1$. It is common practice to set $\eta_t=t^{-\kappa}$, with $\kappa\in (0.5,1]$.

The use of this scheme in MCMC applications consists of setting $H_t$ as a statistic that describes a behavioral aspect of the algorithm at iteration $t$, and $x$ as a tunable parameter of the MCMC algorithm.

For example, in the HMC algorithm, let $\alpha_t$ be the acceptance probability at iteration $t$, $\delta\in (0,1)$ a desired average acceptance probability, $H_t=\delta-\alpha_t$, and $\epsilon$ the step size of the numerical integrator. Then the sequence \ref{DAscheme} with $x=\log\epsilon$  and $h(x)=\mathbb{E}[H_t\mid x]$ is ensured to converge to a value of $x$ such that $h(x)=0$ which translates into a value of  $\epsilon $ such that the expected acceptance probability is equal to $\delta$, $\mathbb{E}[\alpha_t\mid \epsilon]=\delta$. In the examples we develop in later sections we set $\gamma=0.05$, $t_0=10$, $\kappa=0.75$ following the work of \cite{hoffman2014no}.

\subsection{Unadjusted HMC with Stratified Monte Carlo integration}

We now proceed to review the unadjusted HMC algorithm with Stratified Monte Carlo time integration (uHMC) presented in \cite{bou2025unadjusted}. Recall that the HMC algorithm includes a Metropolis–Hastings correction step whose purpose is to remove the bias introduced by the discretization of the Hamiltonian dynamics. In some situations, this bias is tolerable, leading to a scheme known as unadjusted HMC.

In this variant we assume that we wish to sample from a distribution whose density is proportional to $e^{-U(q)}$, where the potential $U$ is continuously differentiable, $K$-strongly convex and its gradient is $L$-Lipschitz. Under these conditions, an alternative scheme for integrating the Hamiltonian dynamics is proposed, which is based on the use of Stratified Monte Carlo. We briefly describe the proposal of this method below.

Let $(q(t), p(t))$ denote the exact flow associated with the Hamiltonian dynamics. Instead of considering a deterministic approximation of this dynamics, we propose a randomized scheme for each integration step. To this end, we fix an evenly spaced time grid given by $\{t_i = ih\}_{i \in \mathbb{N}_0}$, and consider a collection of random variables $\{U_i\}$ such that $U_i \sim \mathrm{Uniform}(t_i, t_{i+1})$. One step of the approximate trajectory is then given by
\begin{gather*}
q_{t_{i+1}} = q_{t_i} + h p_{t_i} + \frac{h^2}{2} F_{t_i}, 
\enskip 
p_{t_{i+1}} = p_{t_i} + h F_{t_i},
\enskip
F_{t_i} = -\nabla U(q_{t_i} + (U_i - t_i)\, p_{t_i}).
\end{gather*}

We highlight that this method uses only one gradient evaluation per step, in contrast to the two gradient evaluations per step required by the leapfrog integrator. Furthermore, \cite{bou2025unadjusted} provides guarantees in the $L^2$-Wasserstein distance by showing that the uHMC kernel is contractive in the $L^2$-Wasserstein distance, implying geometric convergence towards a limiting distribution and stability with respect to the initial condition.

\subsection{Repelling-attracting HMC}
\label{sec:raHMC}

The  repelling-attracting HMC (raHMC) method is designed to sample from multimodal distributions. It is well known in the literature that the Hamiltonian Monte Carlo (HMC) algorithm may encounter difficulties in efficiently exploring target densities with multiple modes (see, e.g., \cite{celeux2000computational,sminchisescu2007generalized}).

The main idea behind raHMC is to modify the Hamiltonian dynamics by incorporating an additional term that models friction in the system, controlled by a friction coefficient. Letting $z_t = (q_t, p_t)$, the dynamics are given by
\begin{gather*}
\frac{d z_t}{dt} = \Omega \nabla_z H(z_t) \pm \Gamma z_t,\\
\text{where}\quad
\Omega :=
\begin{pmatrix}
0 & I \\
- I & 0
\end{pmatrix},
\qquad
\Gamma :=
\begin{pmatrix}
0 & 0 \\
0 & \gamma I
\end{pmatrix},
\end{gather*}
with $\gamma > 0$ denoting the friction coefficient and $I$ the identity matrix of appropriate dimension. When the friction is subtracted, that is, when the dynamics include
$-\Gamma z_t$, the trajectory is attracted toward critical points of $U(q)$. In contrast, when the friction term is added, that is, when the dynamics include
$+\Gamma z_t$, the trajectory is repelled away from the critical points of
$U(q)$. In order to preserve nice properties of HMC such as energy conservation and symplecticity, raHMC combines these two phases as follows.

For a time $T>0$, the time interval is divided into two parts: $(0, T/2]$ and $(T/2, T]$. During the first half, the repelling dynamics is used in order to encourage the chain to move away from the current mode. During the second half, the attracting dynamics is used in order to stabilize the trajectory in a possibly different region of the state space. More precisely, the dynamics are defined by
\begin{equation}\label{raHMCdynamics}
    \frac{d z_t}{dt} =
\begin{cases}
\Omega \nabla_z H(z_t) + \Gamma z_t, & 0 < t \leq T/2, \\
\Omega \nabla_z H(z_t) - \Gamma z_t, & T/2 < t \leq T.
\end{cases}
\end{equation}
The raHMC uses these dynamics to propose the next state of the chain, which is accepted or rejected via a Metropolis--Hastings step with probability
\[
\alpha(z_T \mid z) = \min\left(1, \exp\big(H(\tilde{z}_T) - H(z)\big)\right),
\]
where $\tilde{z}_T=F(z_T)=F(q_T,p_T)=(q_T,-p_T)$, and $z_T$ is the final state of~\eqref{raHMCdynamics}. 

The numerical implementation of the proposed dynamics relies on a conformal Leapfrog integrator. For the first equation in~\eqref{raHMCdynamics}, we have
\[
\Phi_\epsilon^+=\xi_{\epsilon/2}^\Gamma\circ \Phi_\epsilon^H\circ \xi_{\epsilon/2}^\Gamma.
\]
For the second equation in~\eqref{raHMCdynamics}, we have
\[
\Phi_\epsilon^-=\xi_{\epsilon/2}^{-\Gamma}\circ \Phi_\epsilon^H\circ \xi_{\epsilon/2}^{-\Gamma}.
\]
Here, $\Phi_\epsilon^H$ is the usual Leapfrog step, and $\xi_{\epsilon/2}^\Gamma(z)=\xi_{\epsilon/2}^\Gamma(q,p)=(q,e^{\gamma \epsilon/2}p)$ is the exact flow associated with $d z_t/dt=\Gamma z_t$. At time $T=L \epsilon$, $z_T$ is given by 
\[
z_T=\Phi^-_{\epsilon, L/2}\circ\Phi^+_{\epsilon,L/2}(z)=(\Phi^-_\epsilon)^{L/2}\circ (\Phi^+_\epsilon)^{L/2}(z).
\]
The pseudocode for the raHMC is given in Algorithm \ref{raHMC}.

Two parameters must be tuned for the raHMC: the step size $\epsilon$ and the friction coefficient $\gamma$. To apply Dual Averaging in this setting, we consider the log-parametrization $x = (\log \epsilon, \log \gamma)$. The stochastic update is defined using 
\[
h_t = (f_t, f_t),
\quad\text{where}\quad
f_t := \delta - \min\big\{1, \exp\big(H(\Psi_{x,T}(q_t, p)) - H(q_t, p)\big)\big\}.
\]
Here, $\delta \in (0,1)$ is the target acceptance probability and $\Psi_{x,T}(q_t, p)$ denotes the proposal generated by the raHMC dynamics with parameters determined by $x$ at iteration $t$.

This construction allows us to adapt both $\epsilon$ and $\gamma$ simultaneously using a shared acceptance-based feedback signal within the Dual Averaging framework.

\begin{algorithm}[ht]
\caption{The Repelling-Attracting Hamiltonian Monte Carlo (raHMC) algorithm~\cite{vishwanath2024repelling}}
\label{raHMC}
\begin{algorithmic}[1]

\STATE \textbf{Input:} Potential $U := -\log S$, mass matrix $\Sigma \succ 0$, friction coefficient $\gamma > 0$, step size $\epsilon > 0$, length $L > 0$, sample size $N \in \mathbb{Z}_{+}$

\STATE \textbf{Output:} Samples $\{q_1, q_2, \ldots, q_N\} \subset \mathbb{R}^d$

\STATE Initialize $q_0 \in \mathbb{R}^d$ and define
$H(q,p) := U(q) + \frac12 p^\top \Sigma^{-1} p$

\FOR{$n = 1$ to $N$}
    \STATE Sample $p_{n-1} \sim \mathcal{N}(0,\Sigma)$
    \STATE Set $(q,p) \leftarrow (q_{n-1}, p_{n-1})$
    
    \FOR{$i = 1$ to $\lfloor L/2 \rfloor$}
        \STATE $(q,p) \leftarrow \textsc{Conformal-Leapfrog}(q,p,\Sigma,\epsilon,+\gamma)$
    \ENDFOR
    
    \FOR{$i = 1$ to $\lfloor L/2 \rfloor$}
        \STATE $(q,p) \leftarrow \textsc{Conformal-Leapfrog}(q,p,\Sigma,\epsilon,-\gamma)$
    \ENDFOR
    
    \STATE Set $(q,p) \leftarrow (q,-p)$ and $\alpha \leftarrow \min\left\{1,\exp\big(H(q,p)-H(q_{n-1},p_{n-1})\big)\right\}$
    \STATE Sample $u \sim \mathrm{Unif}(0,1)$

    \STATE Update $q_n \leftarrow q\mathbf{1}[u<\alpha]+q_{n-1}\mathbf{1}[u\ge \alpha]$
\ENDFOR

\STATE \textbf{return} $\{q_1,q_2,\ldots,q_N\}$

\end{algorithmic}
\end{algorithm}

\begin{algorithm}[ht]
\caption{Conformal-Leapfrog~\cite{vishwanath2024repelling}}
\begin{algorithmic}[1]

\STATE $\tilde{p} \leftarrow e^{-\gamma \epsilon/2}p - \frac{\epsilon}{2}\nabla U(q)$
\STATE $\tilde{q} \leftarrow q + \epsilon \Sigma^{-1}\tilde{p}$
\STATE $\tilde{p} \leftarrow e^{-\gamma \epsilon/2}\left(\tilde{p} - \frac{\epsilon}{2}\nabla U(\tilde{q})\right)$
\STATE \textbf{return} $(\tilde{q},\tilde{p})$

\end{algorithmic}
\end{algorithm}

\section{Examples and comparisons}
\label{compacompa}

In this section, we empirically evaluate the performance of the algorithms discussed above across five target distributions. To provide a rigorous comparative analysis, we must account for both the statistical quality of the generated samples and the computational cost required to obtain them. 

In Markov Chain Monte Carlo methods, successive samples are inherently correlated, which diminishes the amount of independent information gathered about the target distribution. To quantify this reduction, we utilize the Effective Sample Size (ESS), which estimates the equivalent number of independent and identically distributed draws contained within our autocorrelated chain. A higher ESS indicates superior geometric exploration and more efficient traversal of the target density.

However, evaluating algorithms solely on ESS is insufficient for practical applications. Complex dynamic methods, such as NUTS-DA, may yield a remarkably high ESS by simulating long Hamiltonian trajectories, but at the expense of significant computational overhead per iteration. Therefore, we adopt the computation time normalized by the Effective Sample Size (Time/ESS) as our primary metric of overall efficiency. This ratio explicitly measures the computational cost (in seconds) required to produce exactly one effectively independent sample, providing a holistic evaluation that balances a sampler's statistical robustness with its algorithmic execution time.

Using the Random Walk Metropolis-Hastings (RWMH) as our baseline, we test the manually tuned standard HMC alongside the adaptive NUTS-DA and raHMC-DA algorithms (i.e., NUTS and raHMC with Dual Averaging). We empirically evaluate the performance of all the algorithms discussed above
through five examples. In each experiment, we run the chain for 10,000
iterations. For algorithms that require a warm-up phase, we use 2,000 warm-up
iterations, which are discarded before computing the performance metrics.

We use the \textit{Effective Sample Size} (ESS) to compare the performance of
the algorithms in the univariate examples, and the minimum ESS across
coordinates, denoted by min ESS, in the multivariate examples. The ESS is
computed as
\[
\mathrm{ESS}=\frac{N}{\tau},
\quad\text{where}\quad
\tau = 1 + 2\sum_{k=1}^{\infty}\rho_k,
\]
is the \textit{integrated autocorrelation time}, computed using $\rho_k = \operatorname{Corr}(X_t,X_{t+k})$, and $N$ denotes the number of post-warm-up samples.

In practice, the infinite sum is truncated when the estimated autocorrelation becomes non-positive. In addition to ESS, we compare the algorithms using the Sec./ESS metric, defined as the total running time divided by the ESS. This metric measures the computational cost required to obtain one effectively independent sample. For the bivariate Gaussian mixture presented in Subsection~\ref{mez_gau}, we use a mode indicator that simply records which mean is closer to each sample.

\subsection{Gamma} 
In the first example, a sample is drawn from a $\mathrm{Gamma}(5,1)$ distribution. We ran the algorithms with starting point $x=100$. In order to generate the simulations with the HMC algorithm and its variants, it is necessary to take into account that the kernel of the distribution is given by
\[
S(q)= q^4e^{-q}.
\]

Therefore, the potential energy of the Hamiltonian $H(q,p)$ is given by

\[
U(q)=-4\log(q)+q.
\]

The parameter values used in the experiments were fixed as follows. 
For RWMH, we used a Gaussian proposal with standard deviation $1.0$. 
For HMC, we set $\epsilon=0.1$ and $L=100$, where $\epsilon$ denotes the leapfrog step size and $L$ the number of leapfrog steps. 
For NUTS, we used $\epsilon=0.1$. 
For raHMC, we set $\epsilon=0.08$, $\gamma=0.4$, and $L=20$. 
For uHMC, we used an integration time $T=1.0$ and a discretization step $h=0.02$. 
For HMC-DA, we initialized the step size at $\epsilon_0=0.03$ and set $\lambda=6.0$, whereas for NUTS-DA we used $\epsilon_0=0.03$. 
Finally, for raHMC-DA, we set $\epsilon_0=0.1$, $\gamma_0=0.4$, and $T=10.2$. 
The numerical results are reported in Table~\ref{TablaGamma}, while sample histograms are shown in Figure~\ref{figGamma}.

\begin{table}[t]
\centering
\caption{Results for the Gamma(5,1) distribution}
\label{TablaGamma}
\begin{tabular}{lccc}
\hline
Algorithm & Target Acc. Pr. & ESS & Sec./ESS \\
\hline
HMC & 0.999 & 4802.018 & 1.57e-03 \\

HMC-DA & 0.778 & 6068.682 & 1.43e-04 \\

NUTS & N/A & 456.886 & 5.96e-02 \\

NUTS-DA & N/A & 1408.980 & 2.68e-03 \\

RWMH & 0.847 & 31.200 & 1.31e-02 \\

raHMC & 0.929 & 174.273 & 2.81e-02 \\

raHMC-DA & 0.633 & 1727.601 & 2.25e-02 \\

uHMC & N/A & 74.932 & 1.16e-01 \\

\hline
\end{tabular}
\end{table}

\begin{figure}[t]
\centering

\begin{subfigure}{0.23\textwidth}
    \centering
    \includegraphics[width=\linewidth]{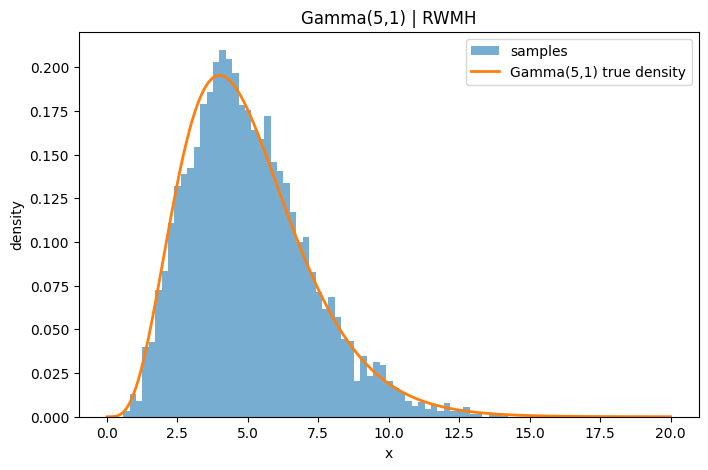}
    \caption{RWMH}
\end{subfigure}
\hfill
\begin{subfigure}{0.23\textwidth}
    \centering
    \includegraphics[width=\linewidth]{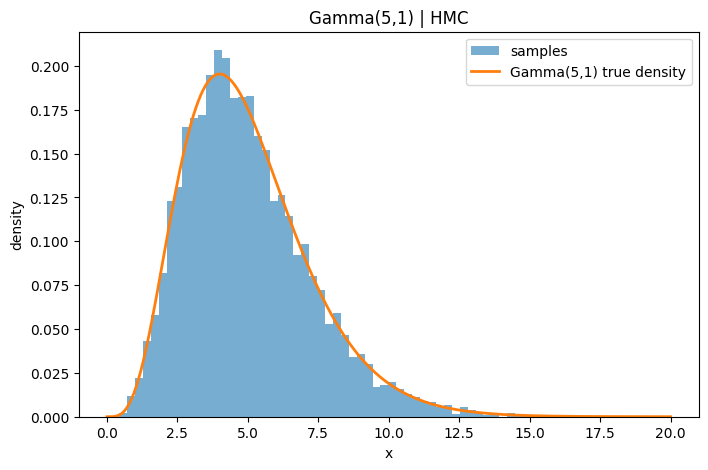}
    \caption{HMC}
\end{subfigure}
\hfill
\begin{subfigure}{0.23\textwidth}
    \centering
    \includegraphics[width=\linewidth]{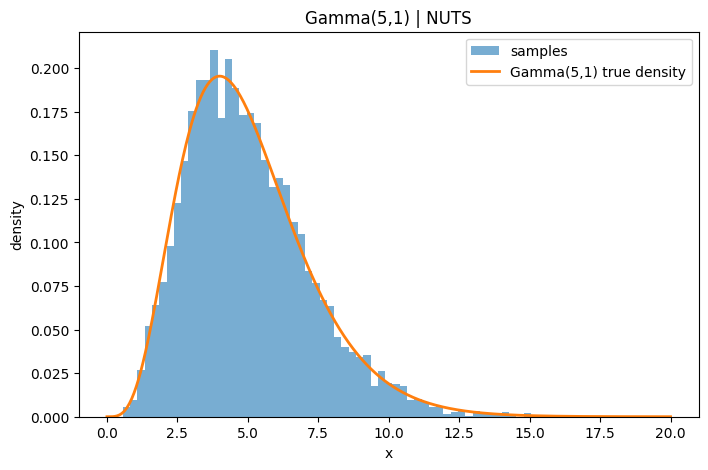}
    \caption{NUTS}
\end{subfigure}
\hfill
\begin{subfigure}{0.23\textwidth}
    \centering
    \includegraphics[width=\linewidth]{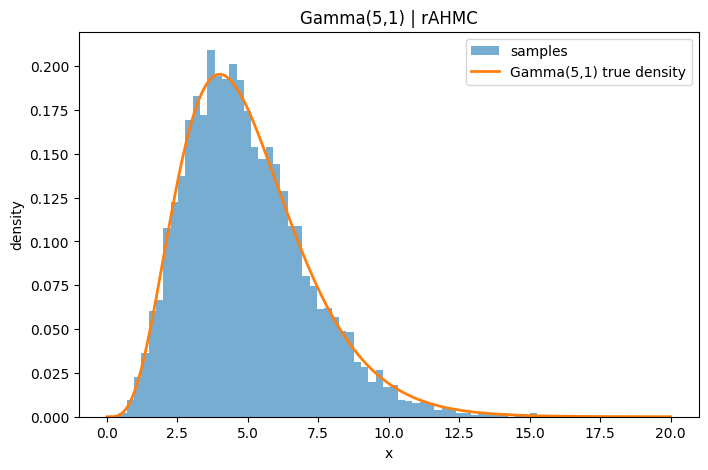}
    \caption{raHMC}
\end{subfigure}

\vspace{0.5cm}

\begin{subfigure}{0.23\textwidth}
    \centering
    \includegraphics[width=\linewidth]{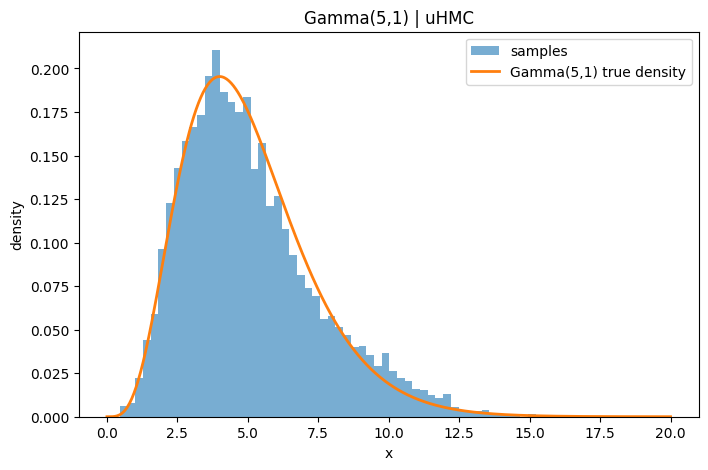}
    \caption{uHMC}
\end{subfigure}
\hfill
\begin{subfigure}{0.23\textwidth}
    \centering
    \includegraphics[width=\linewidth]{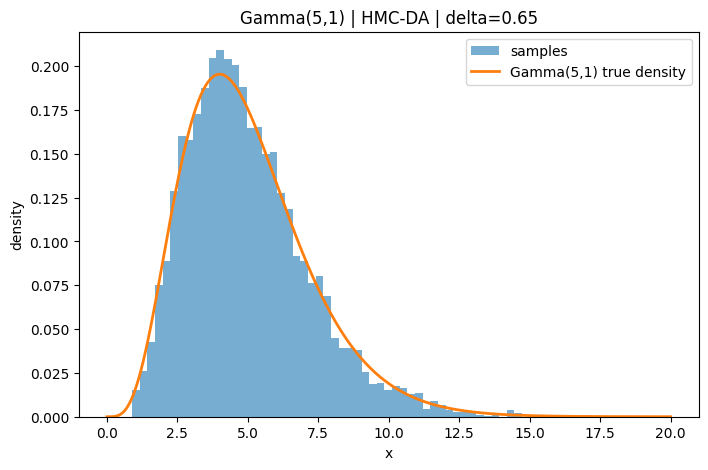}
    \caption{HMC-DA}
\end{subfigure}
\hfill
\begin{subfigure}{0.23\textwidth}
    \centering
    \includegraphics[width=\linewidth]{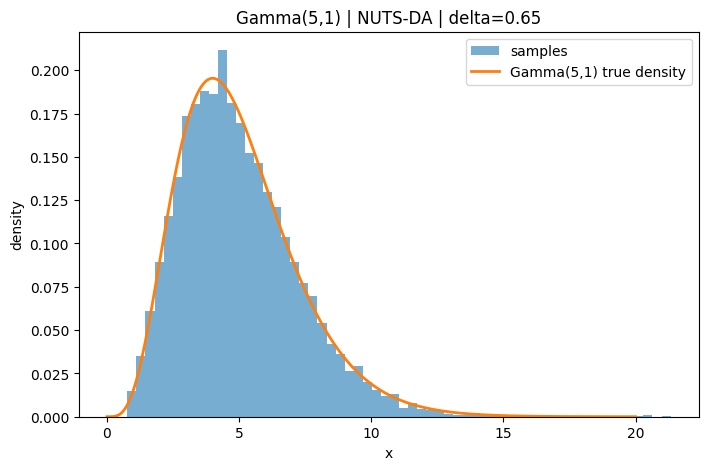}
    \caption{NUTS-DA}
\end{subfigure}
\hfill
\begin{subfigure}{0.23\textwidth}
    \centering
    \includegraphics[width=\linewidth]{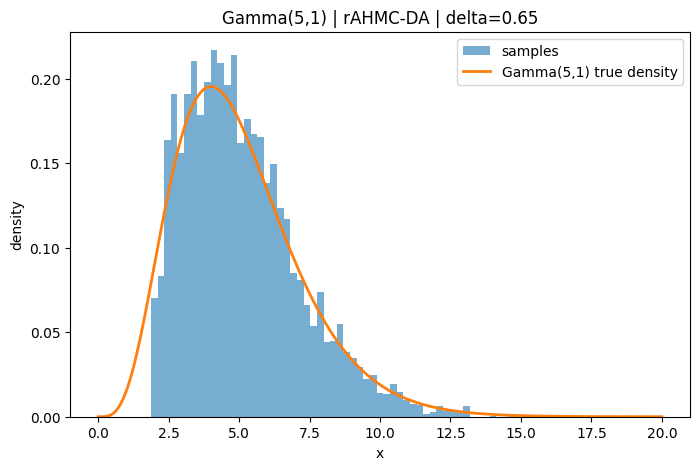}
    \caption{raHMC-DA}
    \label{raHMC-gamma}
\end{subfigure}

\caption{Histograms for the Gamma target distribution.}
\label{figGamma}
\end{figure}

With this configuration, we observe that the highest ESS values are reached by HMC and HMC with Dual Averaging, followed by raHMC and NUTS, both with Dual Averaging. Regarding the Sec./ESS metric, the best performance is attained by HMC with Dual Averaging, while HMC and NUTS also obtain competitive performance. For this example, it seems that HMC is the most competitive algorithm, especially when enhanced by the Dual Averaging scheme. Finally, for the raHMC enhanced by (Figure \ref{raHMC-gamma}), we observe that the left tail is poorly explored; we conjecture that it might be due to the gradient explosion near zero combined with the specific combination of $(\epsilon, \gamma)$ found by the Dual-Averaging scheme.

\subsection{Multivariate normal}

In the second example, a sample is drawn from a $100$-dimensional multivariate Gaussian distribution whose components are independent, have zero mean, and standard deviations $0.01, 0.02, \ldots, 0.99, 1.00$. In all of the algorithms, we generate 10,000 states, with starting point $x_0 = (7, \ldots, 7)$. For the simulations with the HMC algorithm and its variants, we consider that the kernel of the distribution is given by

\[
S(q)= \exp\left(-\frac{1}{2}\sum_{i=1}^{100}\frac{q_i^2}{\sigma_i^2}\right),
\]

where $\sigma_i = 0.01 i$, for $i=1,\ldots,100$.

In turn the potential energy of the Hamiltonian $H(q,p)$ is given by

\[
U(q)= \frac{1}{2}\sum_{i=1}^{100}\frac{q_i^2}{\sigma_i^2}.
\]

The parameter values used in the experiment were fixed as follows. 
For RWMH, we used a Gaussian proposal with standard deviation $0.02$. 
For HMC, we set $\epsilon=0.003$ and $L=300$, while for NUTS we used $\epsilon=0.013$. 
For raHMC, we set $\epsilon=0.003$, $\gamma=0.2$, and $L=300$. 
For uHMC, we used $T=1.5$ and $h=0.003$. 
For the adaptive variants, we initialized HMC-DA with $\epsilon_0=0.013$ and $\lambda=2.0$, NUTS-DA with $\epsilon_0=0.013$, and raHMC-DA with $\epsilon_0=0.009$, $\gamma_0=0.3$, and $T=2.0$. 
The numerical results are reported in Table~\ref{TablaNormal}.

\begin{table}[t]
\centering
\caption{Results for the multivariate ($d=100$) normal distribution}
\label{TablaNormal}
\begin{tabular}{lccc}
\hline
Algorithm & Target Acc. Pr. & Min. ESS & Sec./Min. ESS \\
\hline
HMC & 0.995 & 57.274 & 3.24e-01 \\

HMC-DA & 0.758 & 4.598 & 2.13e+00 \\

NUTS & N/A & 1389.889 & 8.87e-02 \\

NUTS-DA & N/A & 4364.506 & 3.18e-02 \\

RWMH & 0.430 & 2.642 & 9.47e-02 \\

raHMC & 0.133 & 6.368 & 2.54e+01 \\

raHMC-DA & 0.675 & 3.162 & 2.09e+02 \\

uHMC & N/A & 3.804 & 2.07e+01 \\

\hline
\end{tabular}
\end{table}

In this example, we find that the NUTS algorithm seems to be the most efficient for handling high dimensionality, although HMC has better performance than the remaining algorithms.
\subsection{Gaussian Mixture}
\label{mez_gau}

It is well known that the RWMH algorithm encounters difficulties when sampling from multimodal distributions. In order to assess whether the HMC algorithm and its variants exhibit similar limitations, we consider a Markov chain whose target distribution is given by the following Gaussian mixture density:
\begin{equation*}
S(q) = 0.4\, f(x \mid \mu_1, \Sigma_1) + 0.6\, f(x \mid \mu_2, \Sigma_2).
\end{equation*}

Here, $f$ denotes the density of a multivariate normal distribution, defined as
\begin{equation*}
f(x \mid \mu, \Sigma) =
\frac{1}{(2\pi)^{n/2} \det(\Sigma)^{1/2}}
\exp\left(
-\frac{1}{2}(x - \mu)^\top \Sigma^{-1} (x - \mu)
\right).
\end{equation*}
The parameters of the mixture are given by
\begin{equation*}
\mu_1 = (0, 0),
\quad
\Sigma_1 = 
\begin{pmatrix}
1 & 0.5 \\
0.5 & 1
\end{pmatrix}, 
\quad
\mu_2 = (5, 5),
\quad
\Sigma_2 = 
\begin{pmatrix}
1 & -0.3 \\
-0.3 & 1
\end{pmatrix}.
\end{equation*}

We initialized all algorithms at the point $(-1.5,-1.5)$. 
The parameter values used in the experiment were fixed as follows. 
For RWMH, we used a Gaussian proposal with standard deviation $1.0$. 
For HMC, we set $\epsilon=0.20$ and $L=25$, where $\epsilon$ denotes the leapfrog step size and $L$ the number of leapfrog steps. For NUTS, we used $\epsilon=0.20$. For raHMC, we set $\epsilon=0.20$, $L=24$ and friction coefficient $\gamma=0.35$. For uHMC, we used an integration time $T=1.0$ and a discretization step $h=0.02$. For HMC-DA, we initialized the step size at $\epsilon_0=0.20$ and set $\lambda=5.0$. For NUTS-DA, we initialized the step size at $\epsilon_0=0.20$. Finally, for raHMC-DA, we used $\epsilon_0=0.05$, $T=5.0$ and friction coefficient $\gamma_0=0.10$. The numerical results are reported in Table~\ref{TablaNormalBivariada}, while the behavior of the algorithms is illustrated in Figure~\ref{fig:bivariate_gaussian_mixture_histograms}.

\begin{table}[t]
\centering
\caption{Results for the bivariate Gaussian mixture distribution}
\label{TablaNormalBivariada}
\small
\setlength{\tabcolsep}{3pt}
\begin{tabular}{lcccccc}
\hline
Algorithm & Target Acc. Pr. & Min. ESS & Sec./Min. ESS & Switches & Mode 1 & Mode 2 \\
\hline
HMC & 0.997 & 42.050 & 2.65e-01 & 35 & 0.213 & 0.787 \\

HMC-DA & 0.724 & 21.470 & 1.22e-01 & 23 & 0.497 & 0.504 \\

NUTS & N/A & 17.511 & 1.05e+00 & 23 & 0.455 & 0.546 \\

NUTS-DA & N/A & 15.909 & 4.20e-01 & 27 & 0.645 & 0.355 \\

RWMH & 0.526 & 5.320 & 1.03e-01 & 21 & 0.532 & 0.468 \\

raHMC & 0.803 & 120.050 & 2.04e-01 & 188 & 0.368 & 0.632 \\

raHMC-DA & 0.613 & 228.396 & 7.10e-02 & 484 & 0.409 & 0.592 \\

uHMC & N/A & 12.657 & 1.95e+00 & 13 & 0.511 & 0.489 \\

\hline
\end{tabular}
\end{table}
\begin{figure}[t]
\centering

\begin{subfigure}[t]{0.23\textwidth}
    \centering
    \includegraphics[width=\linewidth]{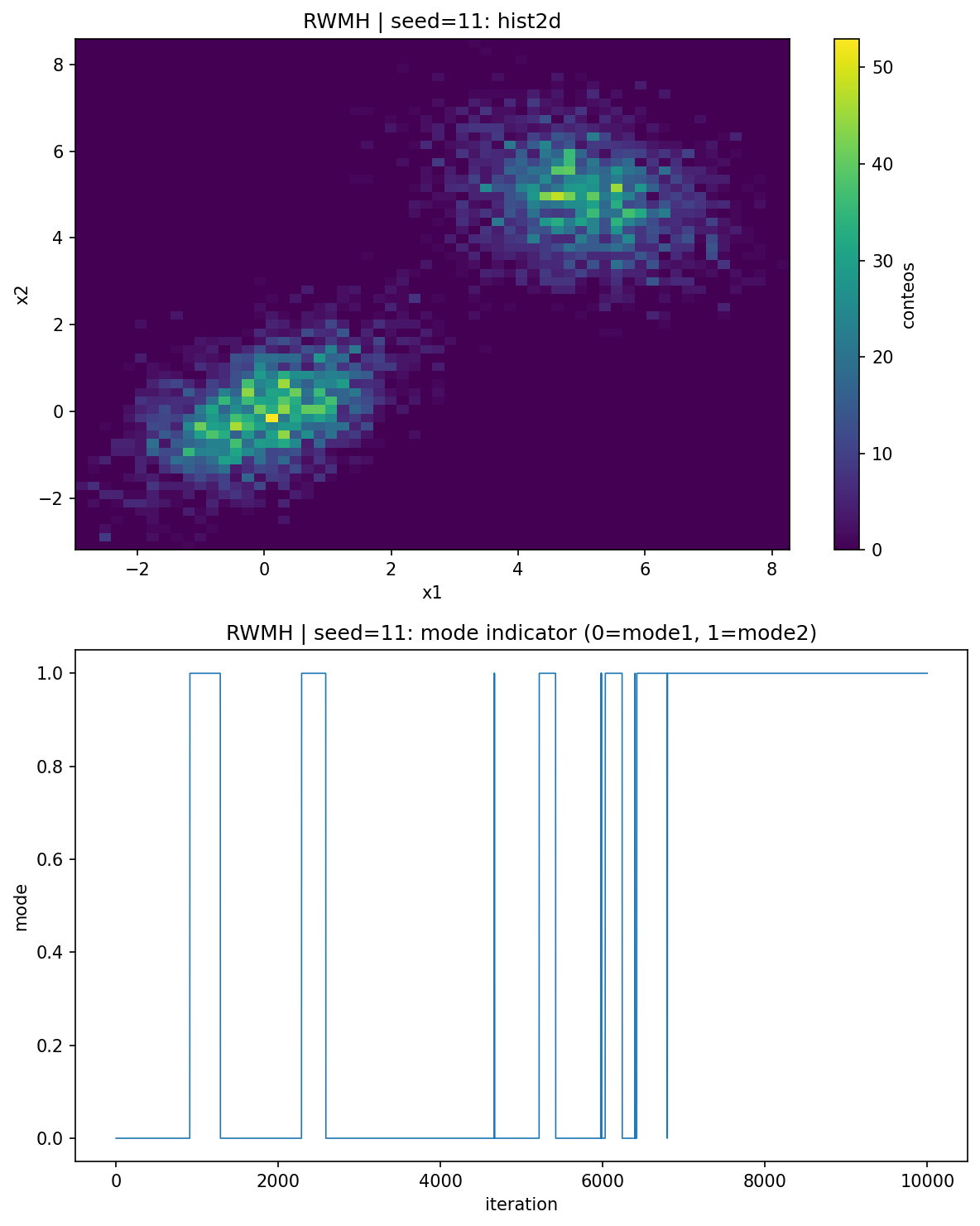}
    \caption{RWMH}
\end{subfigure}
\hfill
\begin{subfigure}[t]{0.23\textwidth}
    \centering
    \includegraphics[width=\linewidth]{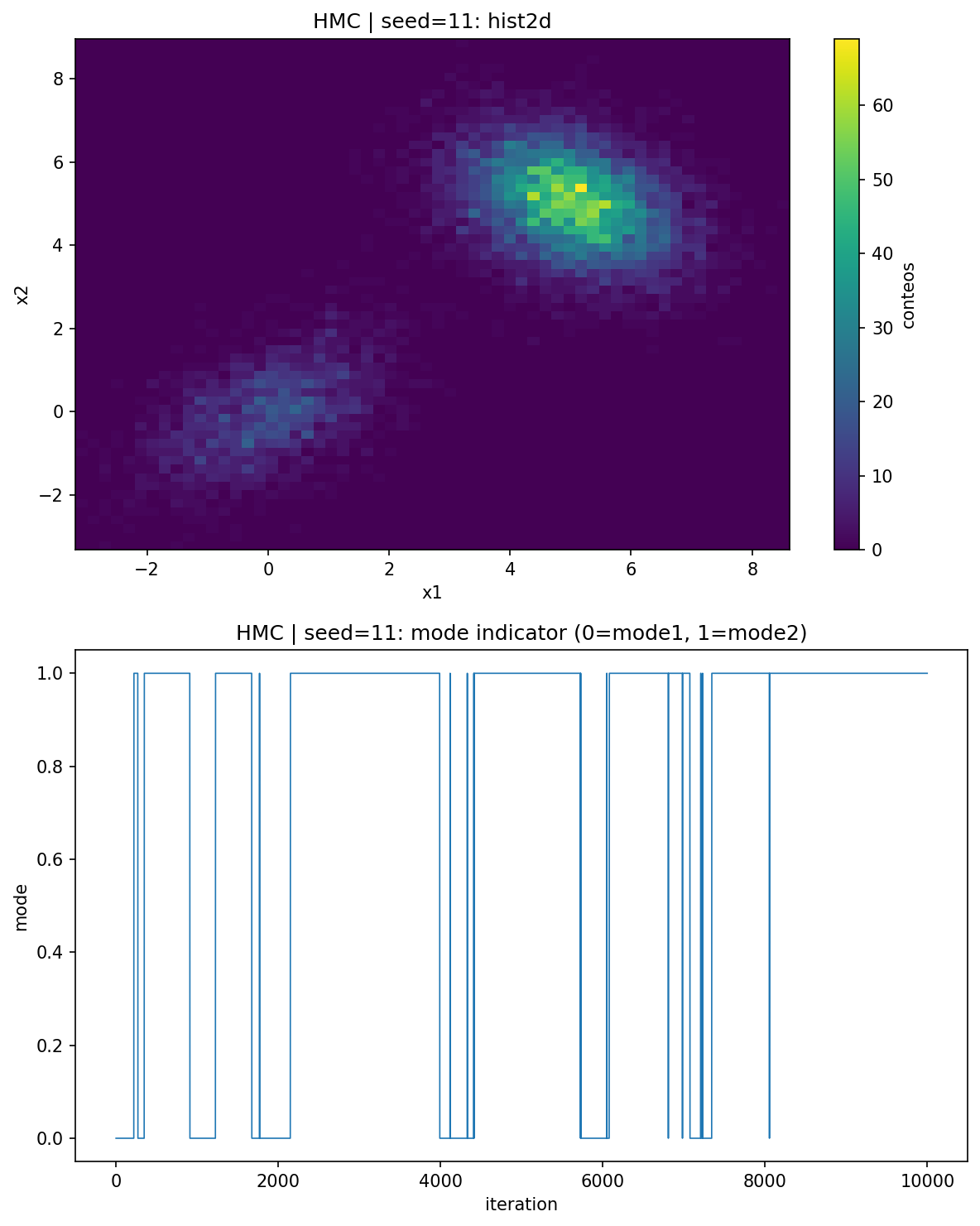}
    \caption{HMC}
\end{subfigure}
\hfill
\begin{subfigure}[t]{0.23\textwidth}
    \centering
    \includegraphics[width=\linewidth]{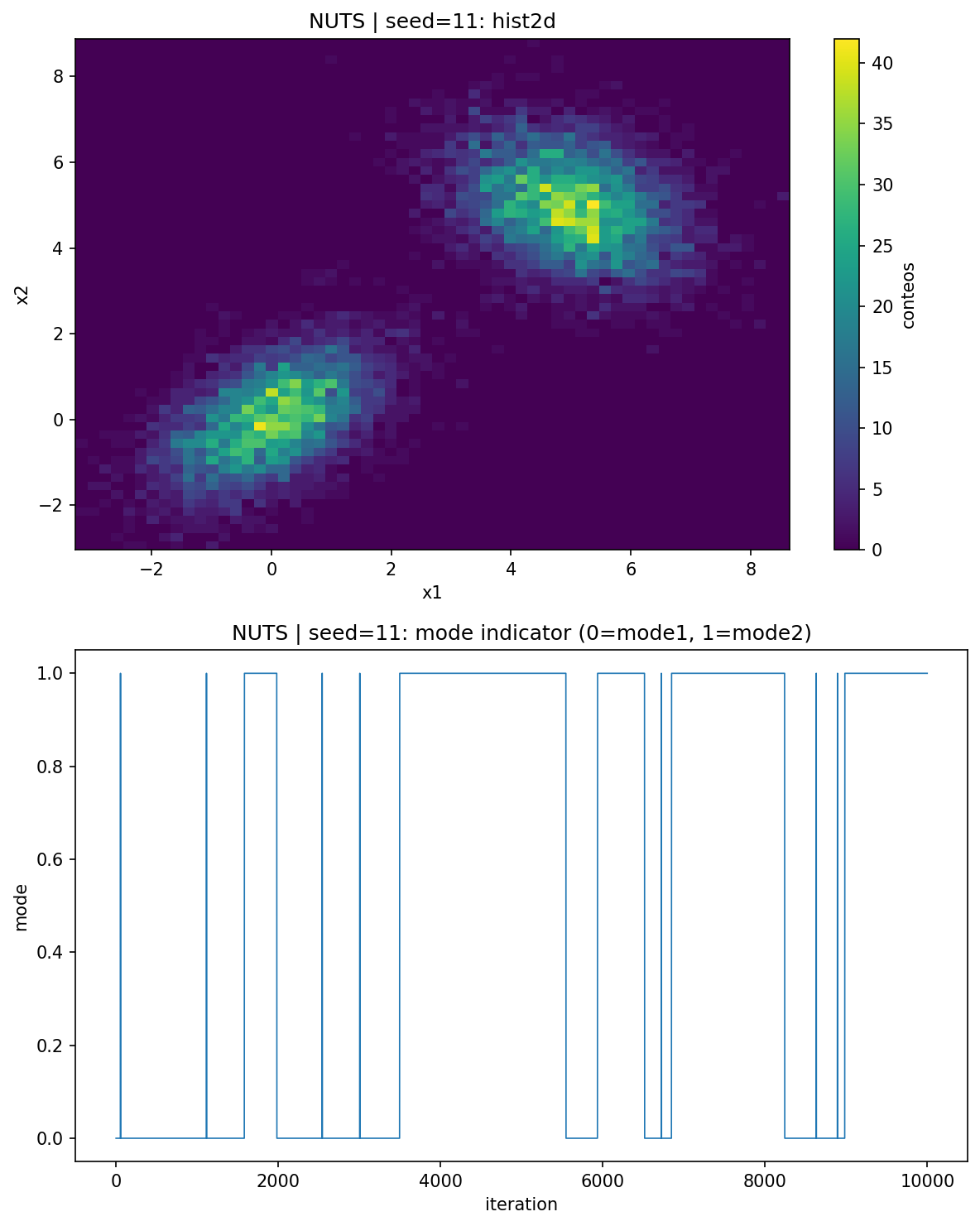}
    \caption{NUTS}
\end{subfigure}
\hfill
\begin{subfigure}[t]{0.23\textwidth}
    \centering
    \includegraphics[width=\linewidth]{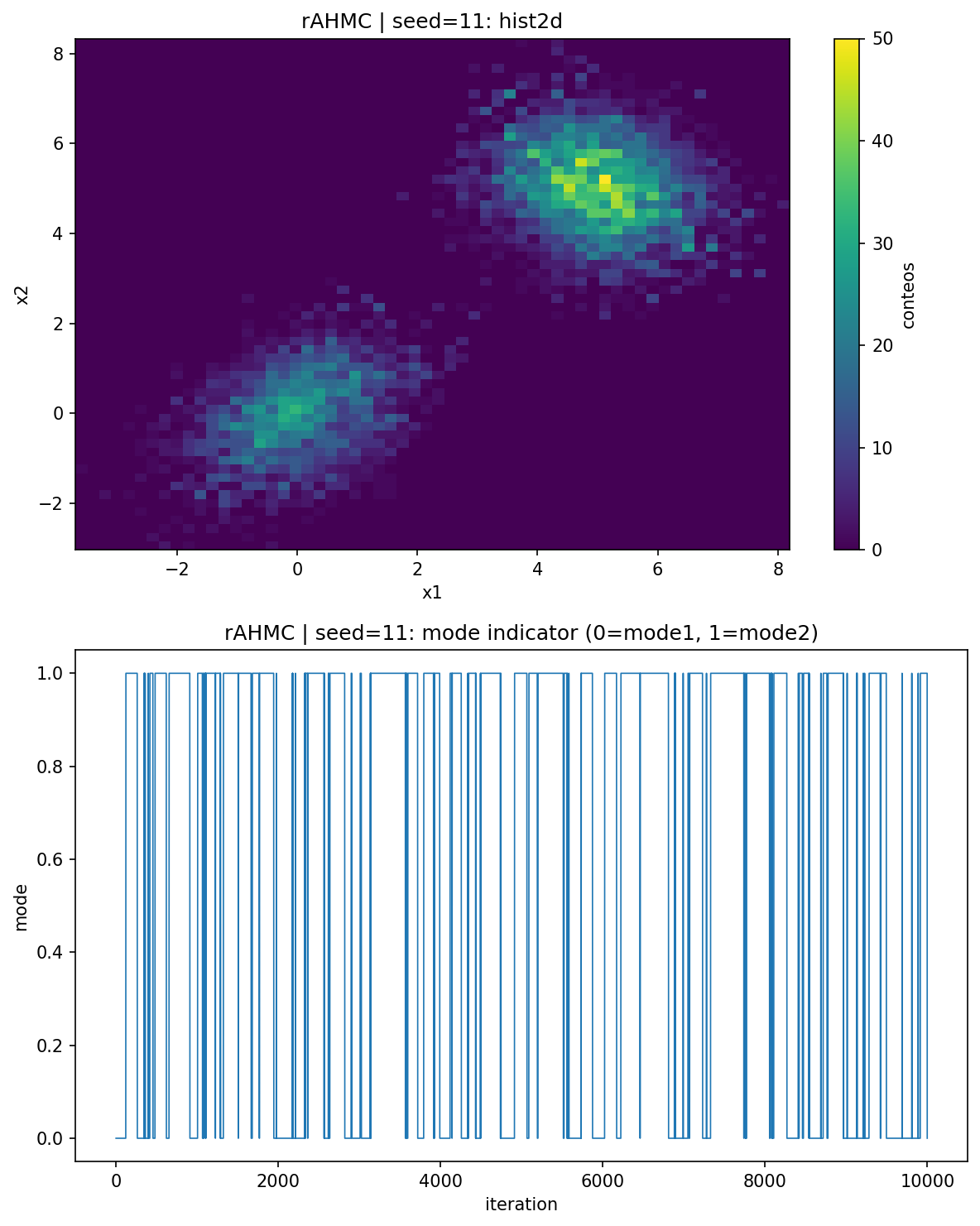}
    \caption{raHMC}
\end{subfigure}

\vspace{0.4cm}

\begin{subfigure}[t]{0.23\textwidth}
    \centering
    \includegraphics[width=\linewidth]{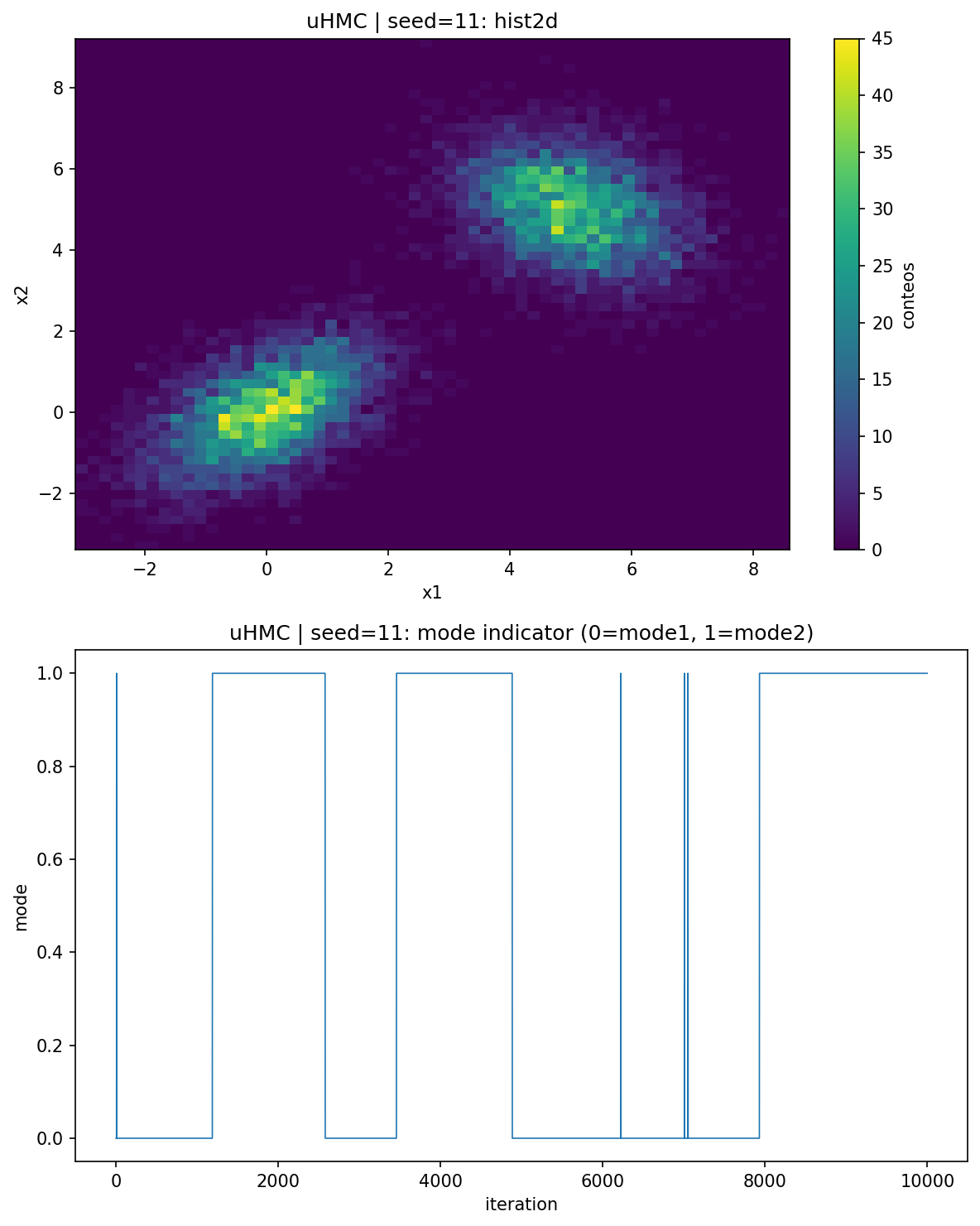}
    \caption{uHMC}
\end{subfigure}
\hfill
\begin{subfigure}[t]{0.23\textwidth}
    \centering
    \includegraphics[width=\linewidth]{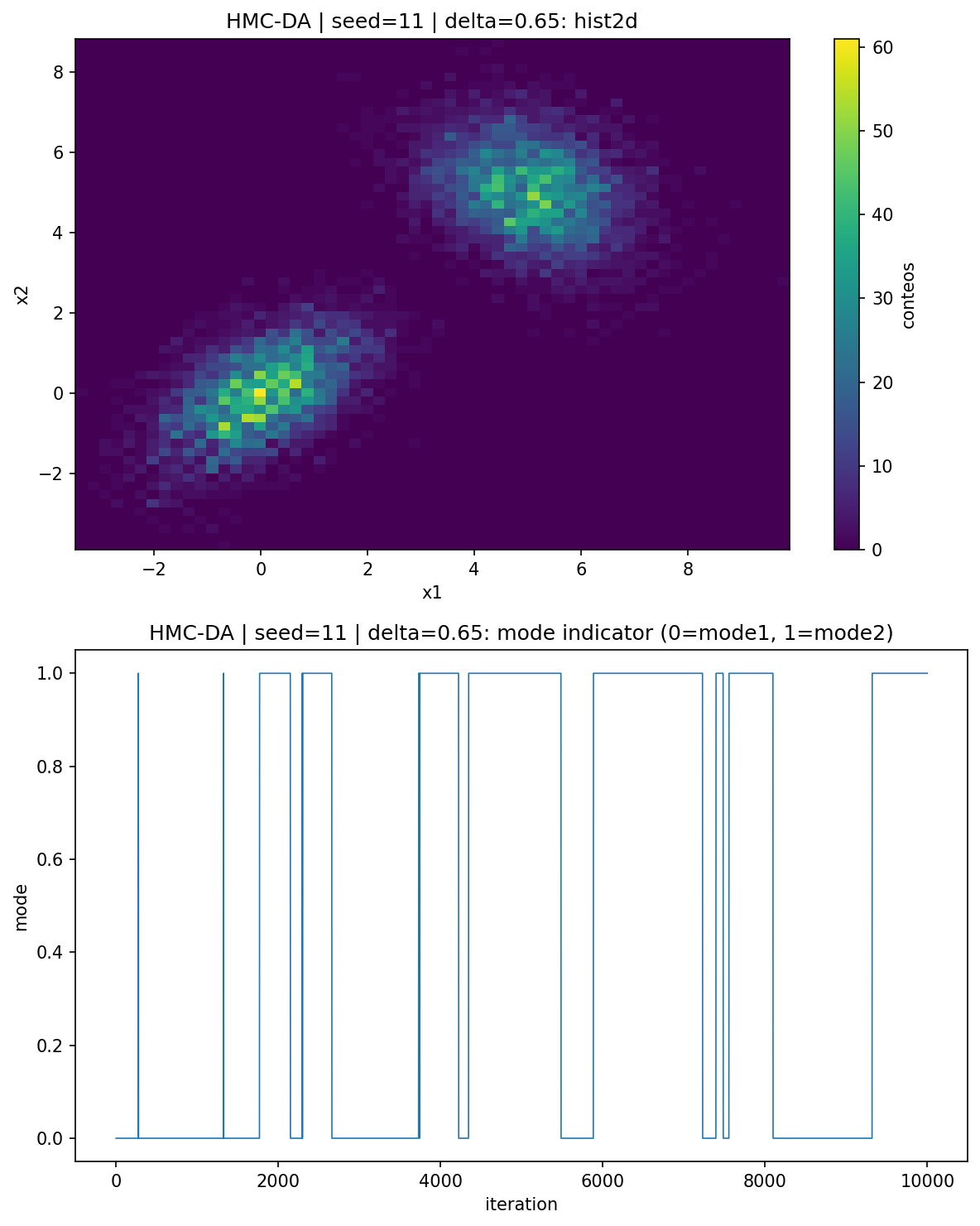}
    \caption{HMC-DA}
\end{subfigure}
\hfill
\begin{subfigure}[t]{0.23\textwidth}
    \centering
    \includegraphics[width=\linewidth]{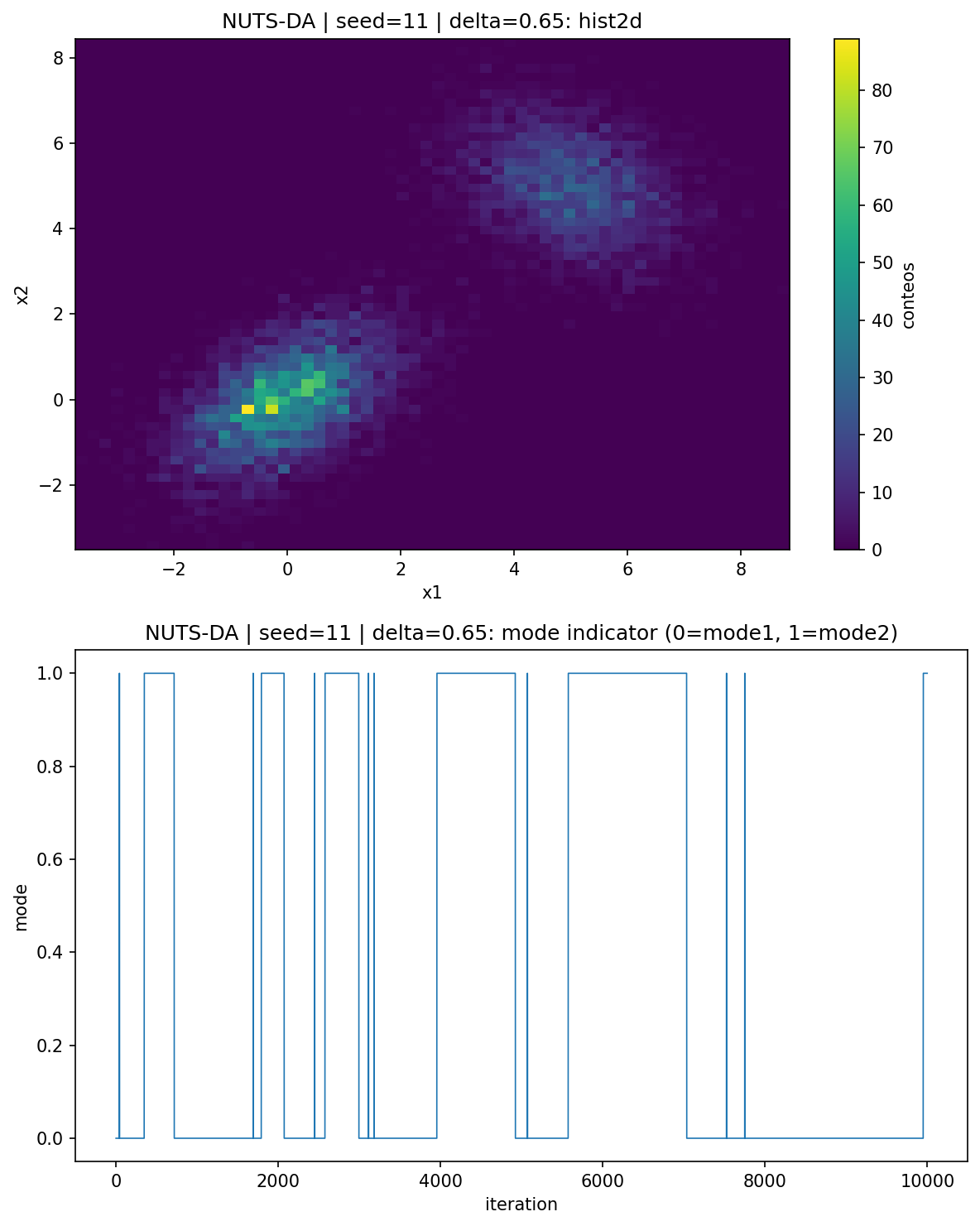}
    \caption{NUTS-DA}
\end{subfigure}
\hfill
\begin{subfigure}[t]{0.23\textwidth}
    \centering
    \includegraphics[width=\linewidth]{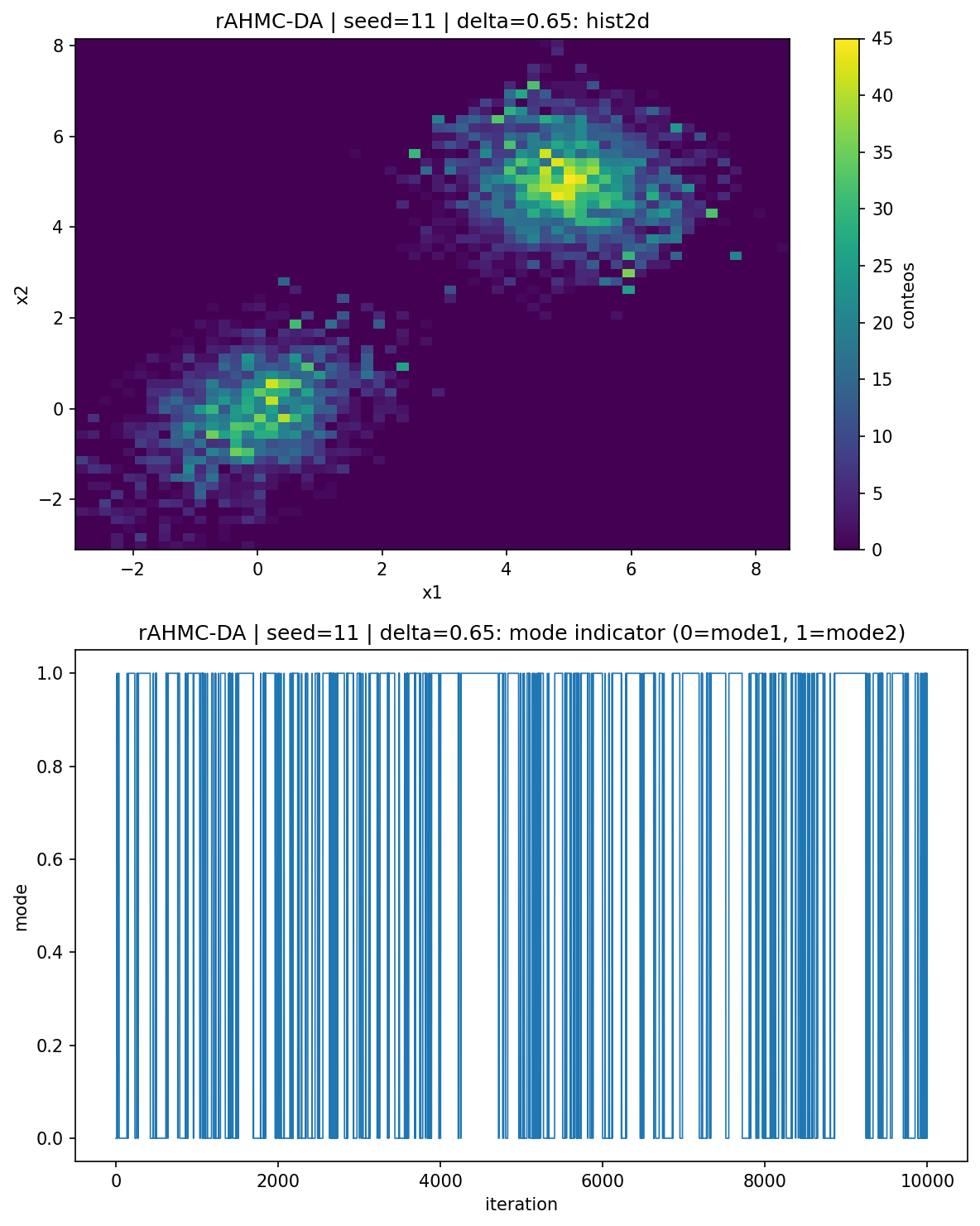}
    \caption{raHMC-DA}
\end{subfigure}

\caption{Histograms for the bivariate Gaussian mixture.}
\label{fig:bivariate_gaussian_mixture_histograms}
\end{figure}

For the bivariate Gaussian mixture example, we note that the raHMC algorithm achieves the best Min. ESS metric. Furthermore, when combined with Dual Averaging, it also achieves the best Sec./Min. ESS metric. This superior performance in a multimodal setting directly illustrates the theoretical advantages of the modified Hamiltonian dynamics introduced in Section \ref{sec:raHMC}. 

While standard HMC and RWMH frequently become trapped in the local energy wells of individual modes, raHMC successfully escapes these regions. By introducing a friction term $\pm\Gamma z_t$ into the dynamics, the algorithm artificially repels the trajectory away from critical points during the first half of the integration and attracts it during the second half. This physical mechanism empowers the sampler to overcome the low-probability energy valleys separating the modes, leading to significantly more frequent mode-switching (484 switches for raHMC-DA) than the remaining algorithms, although it does not spend a perfectly balanced amount of time in both modes on any given run.

\subsection{Hierarchical model}

The example of the eight schools is a Bayesian model presented in \cite{gelman1995bayesian}, where the effect of eight schools' training programmes for the SAT (\textit{Scholastic Aptitude Test}) is studied. This test is used by universities in the USA as a means of assessing the aptitude of applicants,  with the objective of making informed decisions regarding their admission. The SAT is designed to reflect the knowledge acquired over several years of education; therefore, it is anticipated that late efforts to enhance test scores will be ineffective. The objective of the eight-school model is to ascertain the impact of the training programmes on the test scores. Table \ref{escuelas} presents the estimated difference in test scores between students who participated in the training programme and those who did not, along with the respective sample deviations.

\begin{table}[t]
\caption{Training effects on SAT and their sample deviation (cf.~\cite{mescheder2017adversarial}).}
\label{escuelas}
\begin{tabular}{@{}lrrrc@{}}
\hline 
  School & Effect $y_{i}$  & Sample standard \\ &&deviation  $\kappa_{i}^{2}$ \\ 
 \hline 
1 &  2.8 & 0.8\\  
2  & 0.8 & 0.5 \\ 
3  & -0.3 & 0.8 \\ 
4  & 0.7 & 0.6 \\ 
5  & -0.1 & 0.5 \\ 
6  & 0.1 & 0.6 \\ 
7  & 1.8 & 0.5\\ 
8  & 1.2 & 0.4\\ 
 \hline
\end{tabular} 
\end{table}

The proposed model assumes that the training impact of the schools, $y_i \mid \theta_i, \kappa_i$, is conditionally independent and satisfies
\begin{equation*}
y_{i}\mid\theta_{i},\kappa_{i}\sim N(\theta_{i},\kappa^{2})
\quad\text{for}
\quad i\in \left\{1,2,...,8\right\},
\end{equation*}
where $\kappa_{i}$ is known and is given by the sample deviation of the $i$-th school presented in Table \ref{escuelas}, and $\theta_{i}=\mu+\tau\eta_{i}$.

In this exercise, we make inference on the parameters $\mu,\tau  \text{ and } \eta_{i}$ with $i\in \left\{1,2,...,8\right\}$, following the approach presented in \cite{mescheder2017adversarial}, where independent distributions $N(0,1)$ are assigned as the \textit{a priori} distributions for the variables  $\mu,\tau$ and $\eta_{i}$. Put differently, the joint \textit{a priori} density $p$ satisfies 
\begin{equation}\label{priomar}
p\left(\eta, \mu, \tau\right) \propto     \exp\left(-\frac{\mu^{2}}{2}\right)\exp\left(-\frac{\tau^{2}}{2}\right)\stackrel[i=1]{8}{\prod}\exp\left(-\frac{\eta_{i}^{2}}{2}\right),
\end{equation}
where $\eta = \left(\eta_{1},\eta_{2},\ldots,\eta_{8}\right) $. Note that the likelihood $\mathcal{L}\left(  y  \mid \eta, \mu, \tau \right)$ satisfies  
\begin{equation}\label{veromar}
\mathcal{L}\left(  y  \mid \eta, \mu, \tau \right) \propto\prod_{i=1}^{8}\exp\left(-\frac{\left(y_{i}-\left(\mu+\tau\eta_{i}\right)\right)^{2}}{2\kappa_{i}^{2}}\right),
\end{equation}
and the posterior density is given by
\begin{equation}\label{bayesss}
S'\left(y\mid\eta, \mu, \tau\right) = \frac{         p\left(\eta, \mu, \tau\right) \mathcal{L}\left(  y \mid \eta, \mu, \tau \right) }{\int _{(\eta, \mu, \tau) \in  \mathbb{R}^{10}} \mathcal{L}\left(  y \mid \eta, \mu, \tau \right)     p\left(\eta, \mu, \tau\right) d\eta\, d\mu\, d\tau }.
\end{equation}
Substituting \eqref{priomar} and \eqref{veromar} in \eqref{bayesss}, we can observe that the posterior density is proportional to 
\begin{equation}
\begin{aligned}
S(y \mid \eta,\mu,\tau)
&=
\exp\left(-\frac{\mu^2}{2}\right)
\exp\left(-\frac{\tau^2}{2}\right)
\prod_{i=1}^{8}\exp\left(-\frac{\eta_i^2}{2}\right) \\
&\quad \times
\prod_{i=1}^{8}
\exp\left(
-\frac{\left(y_i-(\mu+\tau\eta_i)\right)^2}{2\kappa_i^2}
\right).
\end{aligned}
\label{dis_poste}
\end{equation}

In order to simulate the posterior distribution using HMC, the potential energy of the Hamiltonian is defined as
\begin{eqnarray} \label{potencial_escuela}
U = -\log\left(S\left(y \mid \eta, \mu, \tau\right)\right) = 
\frac{\mu^{2}}{2} + \frac{\tau^{2}}{2} + \sum_{i=1}^{8} \left(\frac{\eta_{i}^{2}}{2} + \frac{\left(y_{i} - \left(\mu + \tau\eta_{i}\right)\right)^{2}}{2\kappa_{i}^{2}}\right),
\end{eqnarray}
and its gradient vector $(\frac{\partial U}{\partial\eta_{1}}, \dots, \frac{\partial U}{\partial\eta_{8}}, \frac{\partial U}{\partial\mu}, \frac{\partial U}{\partial\tau})$ must be computed as follows
\begin{equation}
\begin{gathered}
\frac{\partial U}{\partial \eta_i}
=
\eta_i
-
\left(
\frac{y_i-(\mu+\tau\eta_i)}{\kappa_i^2}
\right)\tau,
\qquad \text{for}\enskip
i=1,\dots,8, \\
\frac{\partial U}{\partial \mu}
=
\mu
-
\sum_{i=1}^{8}
\left(
\frac{y_i-(\mu+\tau\eta_i)}{\kappa_i^2}
\right), \quad
\frac{\partial U}{\partial \tau}
=
\tau
-
\sum_{i=1}^{8}
\left(
\frac{y_i-(\mu+\tau\eta_i)}{\kappa_i^2}
\right)\eta_i.
\end{gathered}
\label{grad_U}
\end{equation}

We initialized all algorithms at the point $(2,2,\ldots,2)$. 
The parameter values used in the experiment were fixed as follows. 
For RWMH, we used a Gaussian proposal with standard deviation $1.0$. 
For HMC, we set $\epsilon=0.05$ and $L=60$, while for NUTS we used $\epsilon=0.05$. 
For raHMC, we set $\epsilon=0.05$, $\gamma=0.35$, and $L=60$. 
For uHMC, we used $T=0.8$ and $h=0.01$. 
For the adaptive variants, we initialized HMC-DA with $\epsilon_0=0.05$ and $\lambda=3.0$, NUTS-DA with $\epsilon_0=0.03$, and raHMC-DA with $\epsilon_0=0.025$, $\gamma_0=0.25$, and $T=3.0$.

To visualize the performance of the algorithms, we display the joint density estimates for $(\mu,\tau)$ and $(\tau,\eta_1)$ in Figure~\ref{fig:eight_schools_results}. 
The corresponding numerical summaries are reported in Table~\ref{TablaEightSchools}.

\begin{table}[t]
\centering
\caption{Results for the Eight schools model}
\label{TablaEightSchools}
\begin{tabular}{lccc}
\hline
Algorithm & Target Acc. Pr. & Min. ESS & Sec./Min. ESS \\
\hline
HMC & 0.994 & 424.688 & 5.95e-02 \\

HMC-DA & 0.681 & 344.303 & 1.75e-02 \\

NUTS & N/A & 331.915 & 2.17e-01 \\

NUTS-DA & N/A & 283.630 & 7.58e-02 \\

RWMH & 0.009 & 17.892 & 2.97e-02 \\

raHMC & 0.414 & 407.254 & 1.45e-01 \\

raHMC-DA & 0.643 & 416.142 & 3.60e-01 \\

uHMC & N/A & 140.613 & 2.94e-01 \\

\hline
\end{tabular}
\end{table}

\begin{figure}[H]
\centering

\begin{subfigure}[t]{0.48\textwidth}
    \centering
    \includegraphics[width=\linewidth]{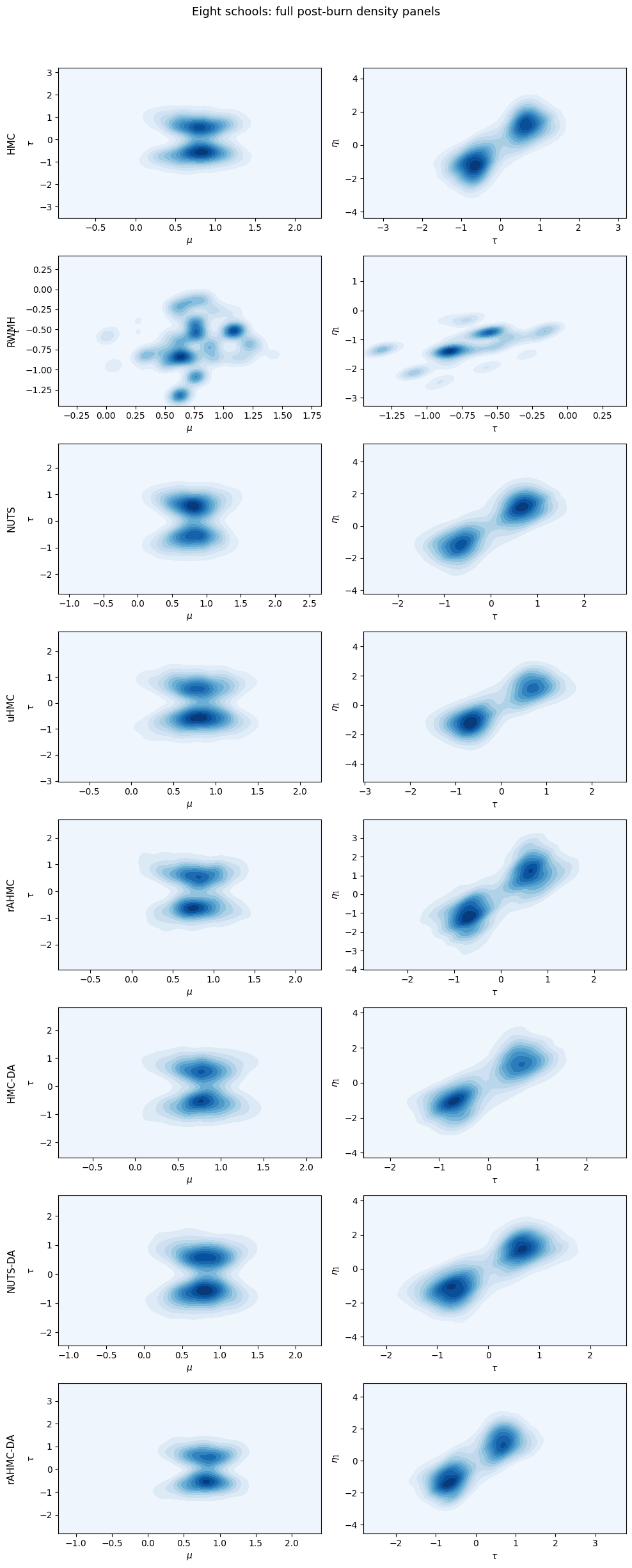}
\end{subfigure}
\hfill
\begin{subfigure}[t]{0.48\textwidth}
    \centering
    \includegraphics[width=\linewidth]{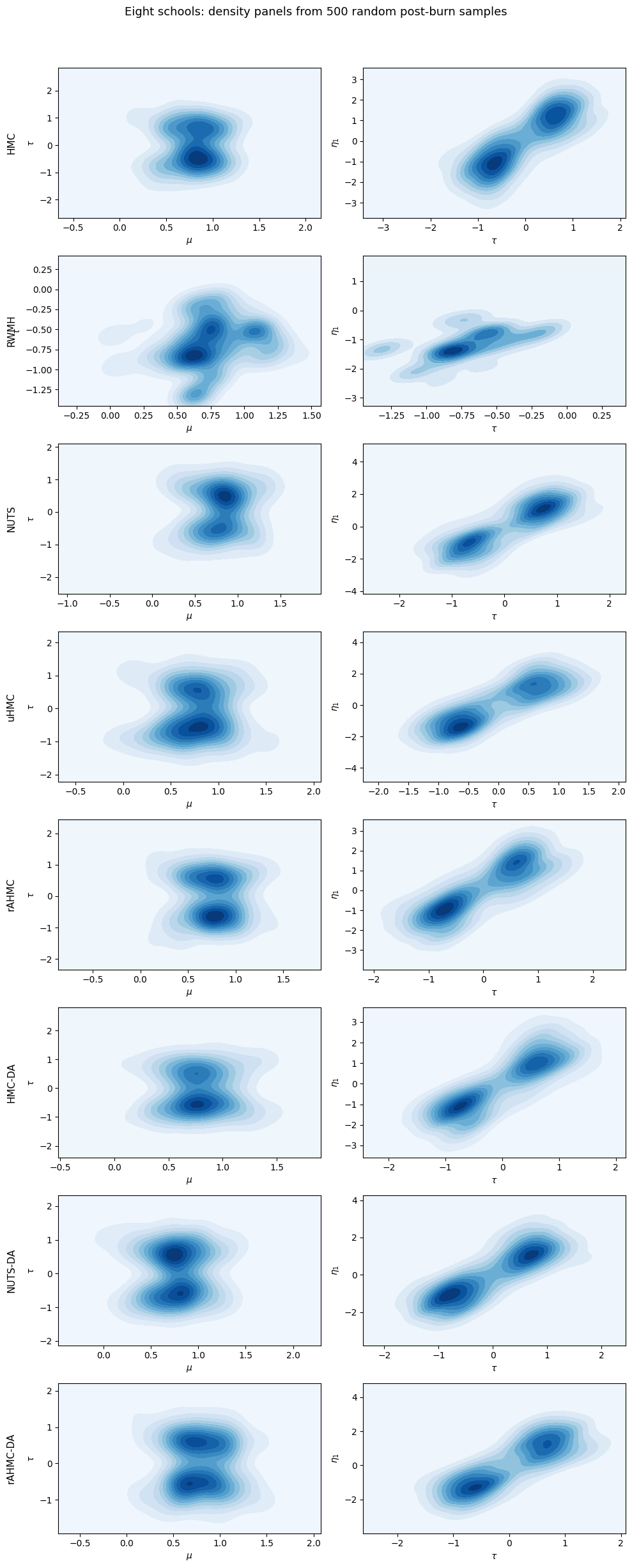}
\end{subfigure}

\caption{Eight schools results.}
\label{fig:eight_schools_results}
\end{figure}
In this example, the algorithms exhibit comparable performance. We only note
that RWMH  have an acceptance rate of 0.009 and a Min. ESS of
17.892 but it is among the fastest algorithms for this problem.
\subsection{Standard Cauchy}

We consider the standard Cauchy distribution as a prototypical example of a heavy-tailed distribution. Due to the absence of finite moments, this setting provides a challenging scenario for assessing the performance of sampling algorithms. 
The parameter values used in the experiment were fixed as follows. 
For RWMH, we used a Gaussian proposal with standard deviation $1.0$. 
For HMC, we set $\epsilon=0.15$ and $L=25$, where $\epsilon$ denotes the leapfrog step size and $L$ the number of leapfrog steps. 
For NUTS, we used $\epsilon=0.15$. 
For raHMC, we set $\epsilon=0.12$, $\gamma=0.5$, and $L=20$. 
For uHMC, we used an integration time $T=1.0$ and a discretization step $h=0.03$. 
For HMC-DA, we initialized the step size at $\epsilon_0=0.15$ and set $\lambda=4.0$. 
For NUTS-DA, we initialized the step size at $\epsilon_0=0.12$. 
Finally, for raHMC-DA, we used $\epsilon_0=0.10$, $\gamma_0=0.5$, and $T=4.0$. 
The numerical results are reported in Table~\ref{Cauchytabla}, while histograms of the generated samples are displayed in Figure~\ref{figCauchy}.

\begin{table}[t]
\centering
\caption{Results for the Cauchy distribution}
\label{Cauchytabla}
\begin{tabular}{lccc}
\hline
Algorithm & Target Acc. Pr. &  ESS & Sec./ESS \\
\hline
HMC & 0.998 & 711.518 & 3.02e-03 \\

HMC-DA & 0.696 & 299.863 & 2.77e-03 \\

NUTS & N/A & 3145.469 & 8.14e-03 \\

NUTS-DA & N/A & 1278.851 & 5.33e-03 \\

RWMH & 0.764 & 109.041 & 1.49e-03 \\

raHMC & 0.767 & 381.063 & 1.82e-02 \\

raHMC-DA & 0.639 & 350.769 & 3.59e-02 \\

uHMC & N/A & 123.394 & 4.23e-02 \\

\hline
\end{tabular}
\end{table}

For this heavy-tailed distribution, we observe that NUTS achieves the best
performance, followed by NUTS with Dual Averaging. In third place, we find HMC,
with an ESS close to that of NUTS with Dual Averaging.
\begin{figure}[H]
\centering

\begin{subfigure}{0.23\textwidth}
    \centering
    \includegraphics[width=\linewidth]{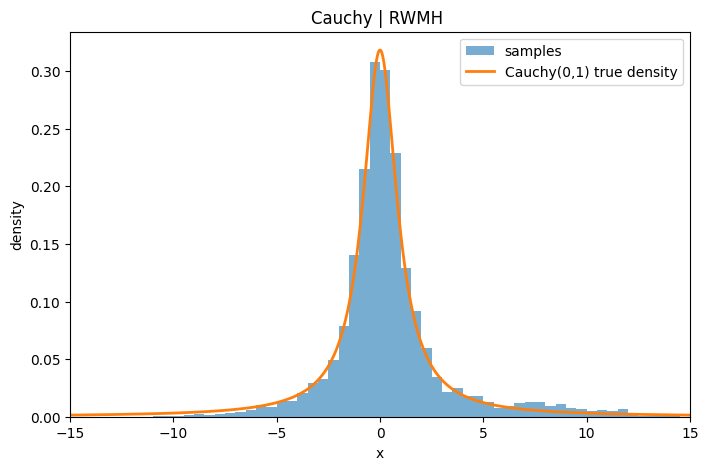}
    \caption{RWMH}
\end{subfigure}
\hfill
\begin{subfigure}{0.23\textwidth}
    \centering
    \includegraphics[width=\linewidth]{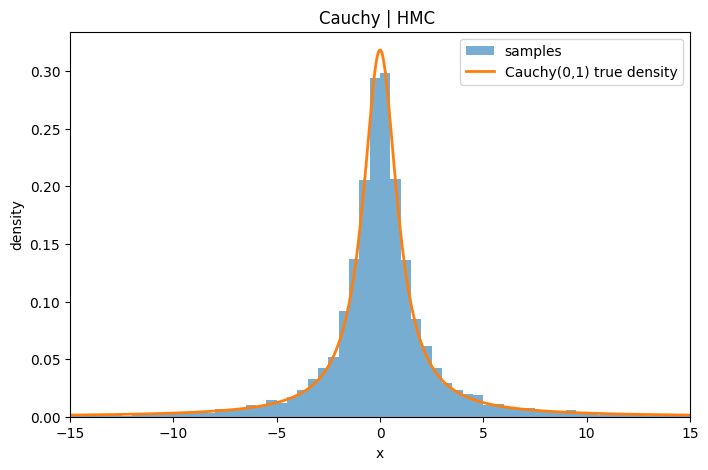}
    \caption{HMC}
\end{subfigure}
\hfill
\begin{subfigure}{0.23\textwidth}
    \centering
    \includegraphics[width=\linewidth]{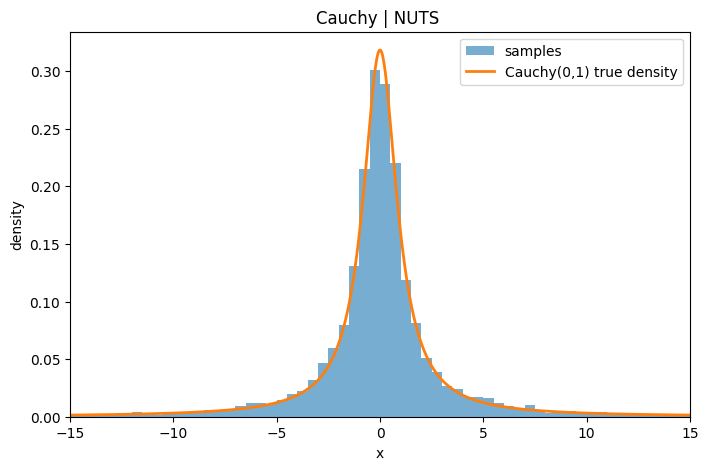}
    \caption{NUTS}
\end{subfigure}
\hfill
\begin{subfigure}{0.23\textwidth}
    \centering
    \includegraphics[width=\linewidth]{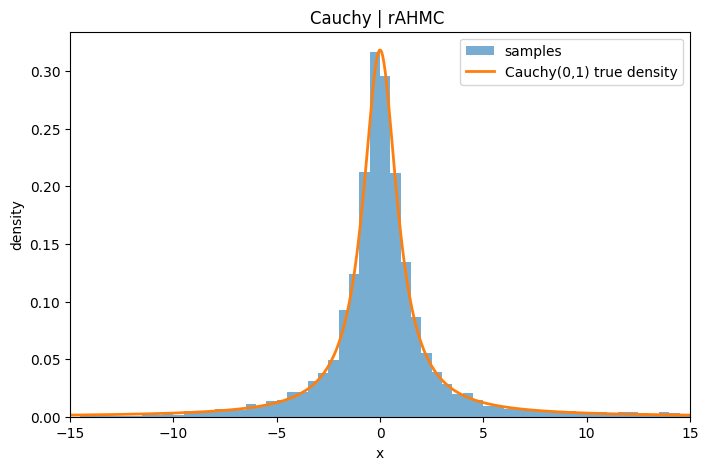}
    \caption{raHMC}
\end{subfigure}

\vspace{0.5cm}

\begin{subfigure}{0.23\textwidth}
    \centering
    \includegraphics[width=\linewidth]{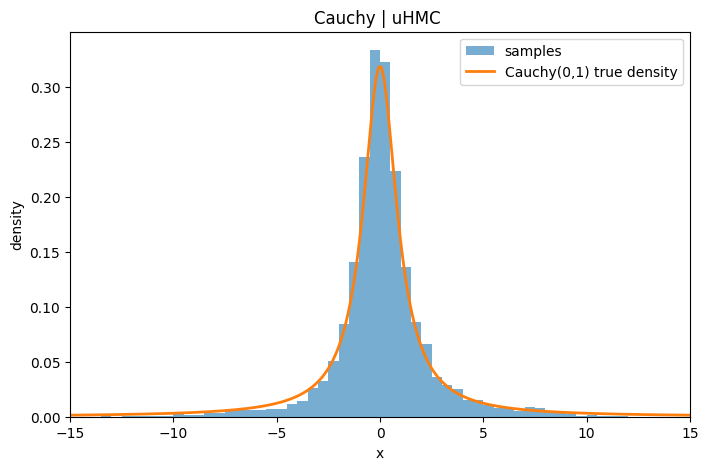}
    \caption{uHMC}
\end{subfigure}
\hfill
\begin{subfigure}{0.23\textwidth}
    \centering
    \includegraphics[width=\linewidth]{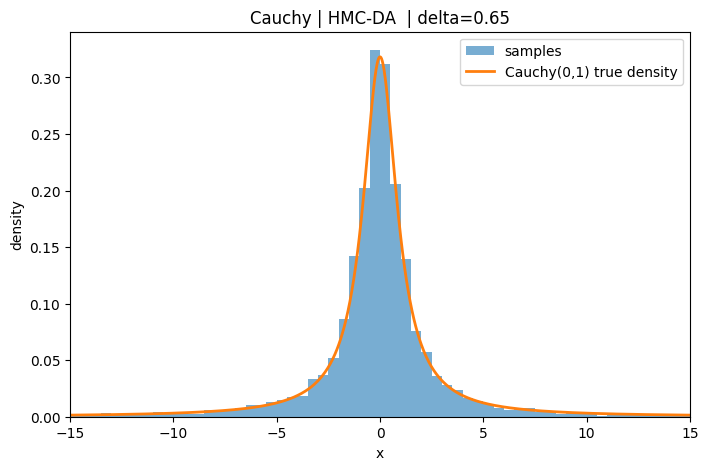}
    \caption{HMC-DA}
\end{subfigure}
\hfill
\begin{subfigure}{0.23\textwidth}
    \centering
    \includegraphics[width=\linewidth]{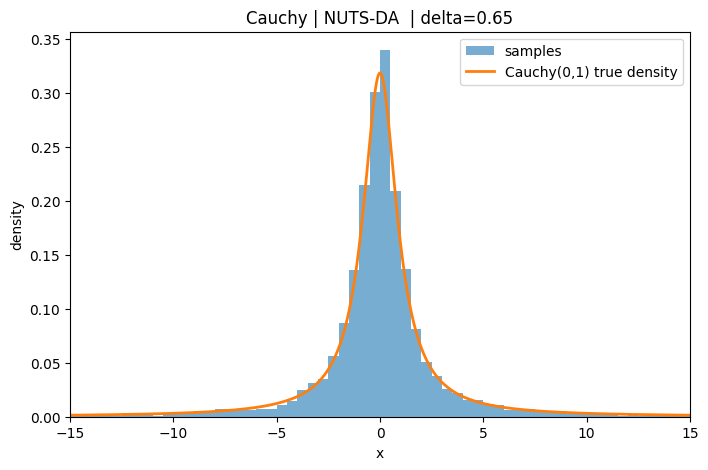}
    \caption{NUTS-DA}
\end{subfigure}
\hfill
\begin{subfigure}{0.23\textwidth}
    \centering
    \includegraphics[width=\linewidth]{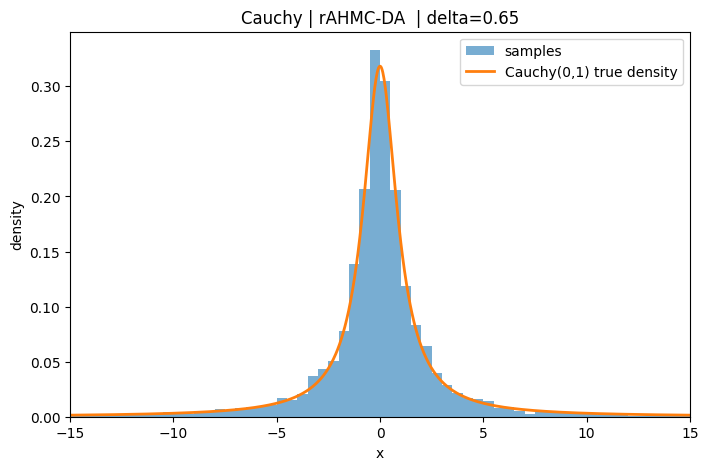}
    \caption{raHMC-DA}
\end{subfigure}

\caption{Histograms for the Cauchy target distribution.}
\label{figCauchy}
\end{figure}

\section{Conclusions}
\label{conclu}

In this work, we have provided a comprehensive exposition of the Hamiltonian Monte Carlo (HMC) algorithm, bridging the gap between its underlying physical principles and its application in Bayesian inference. By introducing the fundamentals of Hamiltonian mechanics in Section~\ref{apena}, we established the geometric and physical intuition necessary to understand the trajectory simulation within the phase space. This theoretical foundation is crucial for grasping how HMC efficiently explores target distributions by suppressing the random walk behavior typical of standard Markov Chain Monte Carlo methods.

Section~\ref{se_construc} detailed the explicit construction of the HMC algorithm, emphasizing the critical role of tuning parameters such as the leapfrog step size and the number of integration steps. To address these tuning challenges, we discussed the Dual Averaging scheme, which provides a robust mechanism for automatically adapting the step size during the burn-in phase, ensuring a stable acceptance rate and reducing the burden of manual parametrization. 

Furthermore, we expanded on the classical HMC framework by incorporating modern variants in Section~\ref{sec:modern_variants}. The No-U-Turn Sampler (NUTS) addresses the problem of selecting the optimal trajectory length by dynamically adapting the number of leapfrog steps, the Dual Averaging adaptively tunes the precision parameter $\epsilon$, while the repelling-attracting HMC (raHMC) modifies the momentum to improve the exploration of complex and multi-modal distributions. 

Through the numerical experiments and comparative analyses presented in Section~\ref{compacompa}, we demonstrated the practical advantages of HMC and its modern extensions over the traditional Random Walk Metropolis-Hastings (RWMH) algorithm. The results, evaluated in terms of execution time, Integrated Autocorrelation Time (IAT), Effective Sample Size (ESS), and acceptance rates, highlight that, while HMC requires computing the gradient of the log-posterior, the resulting efficient exploration of the parameter space heavily compensates for the computational cost per iteration. In particular, the modern variants proved to be highly effective at navigating complex target geometries, further automating the sampling process without sacrificing accuracy.

Ultimately, the theoretical insights and the accompanying Python implementations provide a solid starting point for readers with a mathematical background to fully leverage the potential of Hamiltonian dynamics in complex Bayesian modeling tasks. We end this discussion by stressing a significant limitation of HMC: it necessitates knowledge of the gradient of the target density, which is frequently unavailable, thereby precluding its implementation.


\begin{acks}
JGC was supported by UNAM-DGAPA-PAPIIT grant 36-IA104425, EPSRC grant EP/V009478/1. IB was supported by the grants by Japan Society for the Promotion of Science (21H05200,24K17144) and by the Science and Technology Research Partnership for Sustainable Development grant (JPMJSA2310). MM was supported by SECIHTI scholarship 4023498. 
\end{acks}

\begin{appendix}

\section{Classical Mechanics}\label{appA}

Classical mechanics, also known as Newtonian mechanics,is generally considered to have been initiated in \cite{newton}, where the foundations of classical mechanics are established by three laws of motion and the law of universal gravitation that unifies terrestrial and celestial mechanics. The aim of Newtonian mechanics is to model the motion of particles, which are a mathematical abstraction representing bodies as points in space to which physical properties, such as mass, can be attributed in real and complete three-dimensional vector spaces.

 According to \cite{clas_mec} the motion of a particle  can be described by a vectorial function $r\left( t \right)=\left(x(t),y(t),z(t) \right)$ which indicates the position of the particle at time $t$, as a function of each coordinate axis. The infinitesimal change of the position with respect to time is called velocity, is denoted by $v(t)$, and is calculated as the derivative of the position vector with respect to time, in mathematical notation
\begin{equation}\label{velocity}
v(t)=\dfrac{d r\left(t\right)}{dt}=\left(\frac{dx(t)}{dt},\frac{dy(t)}{dt},\frac{dz(t)}{dt}\right).
\end{equation}

The acceleration at time $t$ is defined as the infinitesimal change of velocity with respect to time. It is calculated as the derivative of the velocity vector and is denoted by $a\left(t\right)$, where
\begin{equation}
a(t)=\dfrac{d v\left(t\right)}{dt}=\left(\frac{d^{2}x(t)}{dt^{2}},\frac{d^{2}y(t)}{dt^{2}},\frac{d^{2}z(t)}{dt^{2}}\right).
\end{equation}

\subsection{Newton's laws of motion}
Newton's laws of motion describe the movement of objects and how forces interact with them. In this section, we will introduce Newton's three laws as presented in  \cite{laws}.

The first law, also called the law of inertia,  states that ``\textit{A body continues to remain in a state of rest or moves with a
uniform speed along a straight line unless subjected to an external
force.}''. This law, therefore, states that for a body to alter its motion, there must be some force that causes it to move, or in mathematical terms, where $F_{i}$ denotes the $i$-th force applied to a body,
\begin{equation}
\stackrel[i=1]{n}{\sum}F_{i}\left(x(t),y(t),z(t)\right)=\left(0,0,0 \right)  \text{ \hspace{0.05 cm}if \hspace{0.05 cm} and  \hspace{0.05 cm} only \hspace{0.05 cm} if \hspace{0.05 cm}}  \frac{dv}{dt}=\left(0,0,0 \right),
\end{equation}
that is, the resultant force applied to a particle is zero if and only if its velocity is not changed.

The second law mentions that in an inertial frame, i.e. a reference frame in which the first law is satisfied, ``\textit{The rate of change of linear momentum of a body is proportional to the magnitude of the external force acting on it and takes place in the direction of that force}.'' Thus, Newton's second law states that in an inertial frame, it is true that for any time $t$,
\begin{equation}
\label{segunda_ley} F\left(x(t),y(t),z(t)\right)=m\left(t \right) \dfrac{d v\left(t\right)}{dt},
\end{equation}
where $m\left(t\right)$ and  $F\left(t\right)$  represent the inertial mass and force at time $t$, respectively. From  \eqref{segunda_ley} follows the equation relating force to mass and acceleration 
\begin{equation}
F\left(x(t),y(t),z(t)\right)=m\left(t \right)a\left(t \right).
\end{equation}

Newton's third law says that: ``\textit{The forces exerted by two bodies on each other are equal in magnitude and opposite in direction.}''  This law is also known as the principle of action and reaction. It states that if one body applies a force on another, the latter will exert a force of equal magnitude and in the opposite direction on the former, i.e.,
\begin{equation}
f_{ij}=-f_{ji},
\end{equation}
where $f_{ij}$ is the force produced by the action of body $i$ on body $j$. 

\subsection{Work and Kinetic Energy}\label{appAB}

In mechanics, a force is said to do work on a particle when it displaces it from one point to another. The definition of work is based on a line integral which is defined in \cite{marsden}.

The work, denoted by $W_{ab}$, in the case of a particle of mass $m$, moving according to a trajectory $C$ under the action of a resultant force $F$ dependent on the particle's position, is given by
\begin{equation}
W_{ab}=\int_{C}F\left(r\left(t\right)\right)\cdotp dr=\sideset{}{_{a}^{b}}\int F\left(r\left(t\right)\right)\cdotp \frac{dr\left(t\right)}{dt}  dt,
\end{equation}
where $r(t)$ is the parameterization of the curve $C$ along which the particle moves and $\frac{dr\left(t\right)}{dt}$, as defined in \eqref{velocity}, is the velocity vector. Using Newton's second law and assuming constant mass, $F=ma\left(t\right)$, it follows that
\begin{eqnarray}
W_{ab}&=&\label{171}\frac{m}{2}(v(b)^{2}-v(a)^{2}).
\end{eqnarray}

If a velocity-dependent function $T$ is defined by $T\left(v\left(t\right)\right)=\frac{1}{2}mv(t)^{2}$,  can be rewritten as  \eqref{171} as
\begin{equation}
\label{traba_cin}   W_{ab}=T\left(v\left(b\right)\right)-T\left( v\left(a\right)\right).
\end{equation}
The term $T(v(t))$ is known as the particle's kinetic energy at time $t$, so  \eqref{traba_cin} states that the work done by a force is equal to the change in kinetic energy in time $b-a$. In mechanics, energy refers to the capacity of bodies to do work, and kinetic energy is that which a body possesses by virtue of being in motion.

\subsection{Conservative Forces and Potential Energy}
\label{cap_potencial}

A conservative force field is one in which the work done by a force to move a particle from point A to point B is independent of the trajectory. By definition, in conservative force fields, for any trajectories $v$ and $w$ from point A to point B, it holds that  
\begin{equation}
\begin{gathered}
\int_{w}F\left(r(t)\right)dr=\int_{v}F\left(r(t)\right)dr\\
\text{or}\qquad
\int_{w}F\left(r(t)\right)dr-\int_{v}F\left(r(t)\right)dr=0.
\end{gathered}
\label{wyyy}
\end{equation}

The left-hand side of the second equality in \eqref{wyyy} can be rewritten as a line integral over a curve Z, defined as a closed path starting and ending at A and passing through point B, that is
\begin{equation}
 \label{int_lineah}   \int_{Z}F\left(r(t)\right)dr=0.
\end{equation}

To find the forces $F$ that satisfy the equality \eqref{int_lineah}, Stokes' theorem is used, which relates surface integrals to line integrals. Its proof can be found in \cite{marsden}. Applying Stokes' theorem to Equation \eqref{int_lineah}, it follows that the conservative forces are those which, for some surface $S$, 
\begin{equation}
 \label{potpot} \int_{S}(\nabla\times F)dS=0.
\end{equation}

Note that  $F=-\nabla U(r(t))$,  where $U$ is a scalar function that depends on the position of the particle, satisfies Equation \eqref{potpot}, because  $\nabla\times(-\nabla U)=0$ as proved in \citep{man}. Therefore, the conservative forces are those that can be written in the form
\begin{equation}\label{wy2}
F=\left(-\frac{\partial U}{\partial x},-\frac{\partial U}{\partial y},-\frac{\partial U}{\partial z}\right)=-\nabla U\left(r\left(t\right)\right).    
\end{equation}
The value  $U\left(r\left(t\right)\right)$ in \eqref{wy2} is known as the potential energy and refers to the energy a particle possesses due to its position.

The mechanical energy of a system of particles at a time $t$, denoted by $E\left(t\right)$, is defined as the sum of the kinetic energy and the potential energy of the system, which is formulated as
\begin{equation}
E\left(t\right)=T\left(v\left(t\right)\right)+U\left(r\left(t\right)\right).
\end{equation}
Since $~\frac{dE\left(t\right)}{dt}=0$, as shown in \cite{clasica_m}, the mechanical energy is conserved, i.e., it does not change over time.
\end{appendix}


\bibliographystyle{imsart-number} 
\bibliography{bibliografia.bib}       

\end{document}